\documentclass[aoas,preprint]{imsart}

\usepackage{amsthm,amsmath,amsfonts}
\RequirePackage[colorlinks,citecolor=blue,urlcolor=blue]{hyperref}
\usepackage[round]{natbib}
\usepackage[english]{babel}
\usepackage{latexsym}
\usepackage{subfigure}
\usepackage{epsfig}
\usepackage{color,soul}
\usepackage{moreverb}
\usepackage{mathtools}
\usepackage{tikz, color}
\usepackage{enumerate,linegoal}
\usepackage{calc}
\usepackage{latexsym}
\usepackage{amsmath}
\usepackage{amsfonts}
\usepackage{amssymb}
\usepackage{bm}
\usepackage{psfrag}
\usepackage{graphicx}
\usepackage{epstopdf}
\usepackage{framed}
\usepackage{booktabs}
 \usepackage[flushleft]{threeparttable}

\usepackage{xfrac,array,booktabs}
\usepackage{verbatim}
\usepackage{enumitem}
\usepackage{pdflscape}
\usepackage{rotating}

\usetikzlibrary{matrix}

\startlocaldefs
\newcommand{\expect}{\operatorname{E}}
\newcommand{\var}{\operatorname{var}}

\renewcommand{\eqref}[1]{(\ref{#1})}

\newcommand{\tp}{^{\top}}

\newtheorem{proposition}{Proposition}
\newtheorem{corollary}{Corollary}
\endlocaldefs

\begin{document}

\begin{frontmatter}

\title{Quantile contours and allometric modelling for risk classification of abnormal ratios with an application to asymmetric growth-restriction in preterm infants}
\runtitle{Quantile contours and allometric modelling}

\begin{aug}
\author{\fnms{Marco} \snm{Geraci}\corref{}\thanksref{t1,m1}\ead[label=e1]{geraci@mailbox.sc.edu}}
\author{\fnms{Nansi S.} \snm{Boghossian}\thanksref{t1}\ead[label=e2]{nboghoss@mailbox.sc.edu}}
\author{\fnms{Alessio} \snm{Farcomeni}\thanksref{t2}\ead[label=e3]{alessio.farcomeni@uniroma1.it}}
\and
\author{\fnms{Jeffrey D.} \snm{Horbar}\thanksref{t3,t4}\ead[label=e4]{horbar@VTOXFORD.org}}

\thankstext{t1}{Corresponding author: Marco Geraci, Department of Epidemiology and Biostatistics, Arnold School of Public Health, University of South Carolina, 915 Greene Street, Columbia SC 29209, USA. \printead{e1}}

\thankstext{m1}{Supported by an ASPIRE grant from the Office of the Vice President for Research at the University of South Carolina and by the National Institutes of Health -- National Institute of Child Health and Human Development (Grant Number: 1R03HD084807-01A1).}

\runauthor{Geraci et al}

\affiliation{University of South Carolina\thanksmark{t1}, Sapienza--University of Rome\thanksmark{t2}, University of Vermont\thanksmark{t3}, and Vermont Oxford Network\thanksmark{t4}}

\end{aug}

\begin{abstract}
\quad We develop an approach to risk classification based on quantile contours and allometric modelling of multivariate anthropometric measurements. We propose the definition of \textit{allometric direction} tangent to the directional quantile envelope, which divides ratios of measurements into half-spaces. This in turn provides an operational definition of directional quantile that can be used as cutoff for risk assessment. We show the application of the proposed approach using a large dataset from the Vermont Oxford Network containing observations of birthweight (BW) and head circumference (HC) for more than 150,000 preterm infants. Our analysis suggests that disproportionately growth-restricted infants with a larger HC-to-BW ratio are at increased mortality risk as compared to proportionately growth-restricted infants. The role of maternal hypertension is also investigated.
\end{abstract}



\end{frontmatter}

\section{Introduction}
\label{sec:1}
The remarkable works on anthropometry by Adolphe Quetelet and Sir Francis Galton in the 19th century gave birth to a new field of scientific investigation within which the medical and statistical sciences developed a long-lasting and profitable collaboration. In turn, this has given rise to countless research studies in public health and to the development of important analytic methods. For example, the body mass index (BMI), also known as the \emph{Quetelet index}, is universally applied by researchers and clinicians to classify individuals into categories such as `underweight', `overweight', and `obese' as these may be at higher risks of poorer health outcomes. Classification thresholds for these categories are defined as percentiles of the BMI distribution in a reference population and are published by public health institutes like the World Health Organization (\url{http://www.who.int/childgrowth/standards/en/}).

The BMI index is a well-known example of ratio of anthropometric variables (mass/height$^2$). Another similar ratio is the corpulence or Rohrer index (mass/height$^3$). Indeed, there exists a plethora of indices where body measurements are combined as ratios, often upon power transformations, where the scaling depends on the relationship between the variables involved in the ratio. If the correct scaling exponent is applied, then no residual association should be observed between the index and the scaling variable \citep{Heymsfield2007}. Otherwise, differential misclassification bias may result when the goal is to assess risk in different categories of the anthropometric index \citep{DPCG2005}.

In this paper, we are specifically interested in anthropometric measurements for very preterm infants (22 to 29 weeks' gestation). Preterm babies, particularly those born at lower gestational ages, have high risks of mortality, morbidities, and neurodevelopmental impairment \citep{Stoll2010,Horbar2012}. For example, it is estimated that at 22 and 23 weeks' gestation the mortality rate can be as high as 80\% \citep{Stoll2010}. At these ages, there are significant rates of respiratory distress syndrome (94\%), patent ductus arteriosus (46\%), severe intraventricular hemorrhage (16\%), necrotizing enterocolitis (11\%), and late-onset sepsis (36\%) \citep{Stoll2010}.

Preterm birth is not the only risk factor. Growth restriction, usually defined as birthweight (BW) less than the 10th percentile for gestational age (GA)---or small for gestational age (SGA)---further raises already high risks among preterm infants \citep{Bernstein2000} and, hence, is used as an indicator for secondary and tertiary prevention of mortality and adverse outcomes. The etiology of SGA is multifactorial with some causes linked to, for example, smoking, placental insufficiency, environmental factors, and maternal complications like preeclampsia. These factors not only impact BW but might also affect the size of the head (as measured by head circumference, HC, right after birth), with consequences that may vary according to the particular period of pregnancy in which the insult has occurred. It has been theorized that: if the insult occurs early during pregnancy or even before conception, growth restriction is symmetric (or proportional) and both BW and HC are affected; if the insult occurs later during pregnancy, growth restriction is asymmetric (or disproportional), with negative consequences mainly for BW, thus resulting in a larger HC-to-BW ratio \citep{Vandenbosche1998,Saleem2011}. Asymmetric growth is thought to be an adaptive mechanism that is put in place to protect the brain. In response to placental insufficiency, which is often caused by hypertension and leads to intrauterine growth restriction, the fetus adapts its circulation to preserve oxygen and nutrient supply to the brain (the `brain-sparing' effect). Some studies investigated the determinants of fetal growth and body proportionality, as well the effect of the latter on neonatal outcomes. They found that (severe) pregnancy-related hypertension, which is the development of new hypertension after 20 weeks' gestation, is strongly associated with a larger HC-to-BW ratio \citep{Kramer1990a}. The latter, in turn, was found to be a risk factor for stillbirth and fetal distress \citep{Kramer1990b}.

\begin{figure}
\includegraphics[scale = .40]{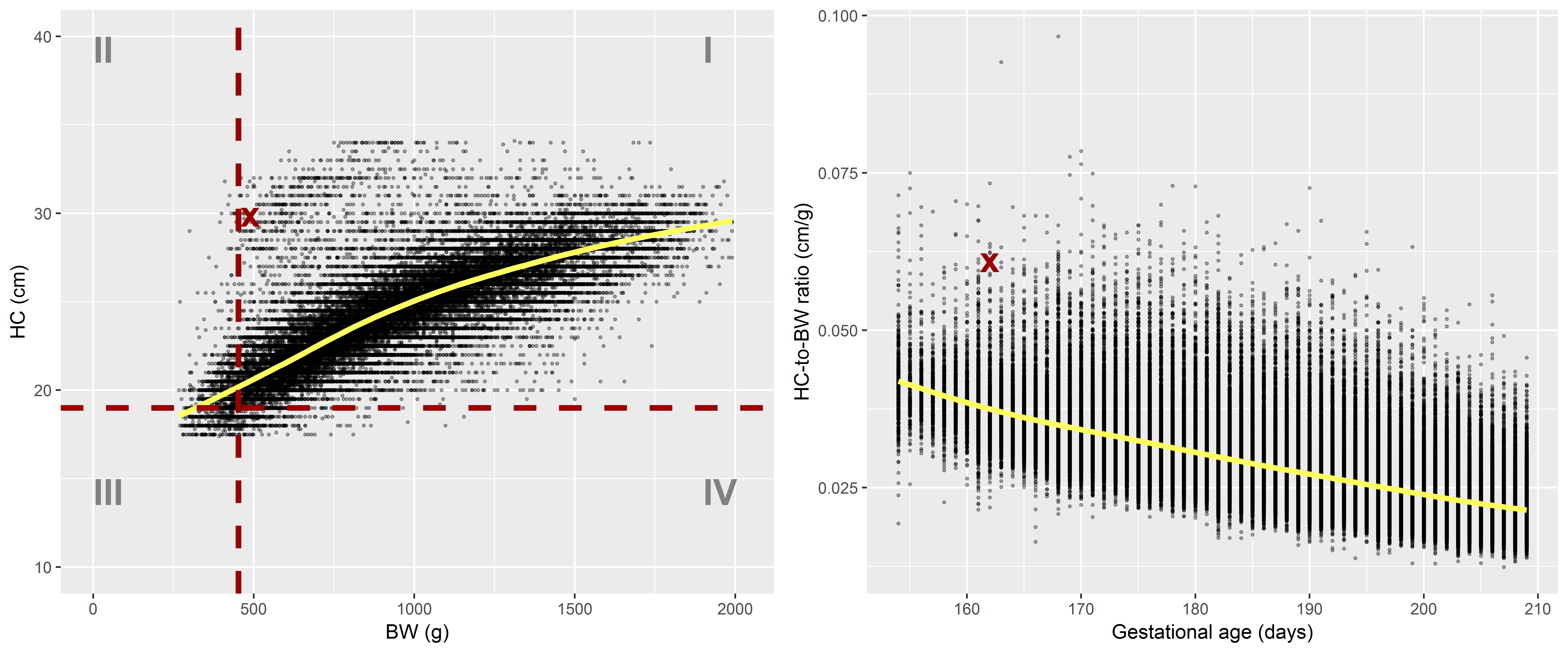}
\caption{Left: head circumference (HC) vs birthweight (BW). The horizontal and vertical dashed lines mark, respectively, the estimated 10th percentile of HC and BW at age 162 days, and divide the plane into four quadrants numbered I to IV. Right: HC-to-BW ratio vs gestational age. Superimposed solid lines represent smoothing splines predictions.\label{fig:1}}
\end{figure}

Thus, it is natural to wonder whether there are differences in terms of health outcomes between infants with unusual HC-to-BW ratio and those with normal HC-to-BW ratio. The abnormality of a ratio reflects the abnormality of either the numerator or denominator, relative to the other. This is illustrated in Figure~\ref{fig:1} which shows the relationship between BW and HC in preterm infants using data from the Vermont Oxford Network (VON), a large network of medical centers. The observation marked by a thick cross in the left plot of Figure~\ref{fig:1} (a girl born at 162 days of gestation) presents an abnormal HC-to-BW ratio for her GA (right plot), a possible consequence of the brain-sparing effect. Note that, individually, the BW and HC measurements for this girl do not present a concern since they both lie above the respective cutoffs for small BW and small HC.

The characterization of abnormal bivariate observations require the application of multivariate approaches to joint ranking. Multivariate modelling has a long tradition in parametric statistics. Models for joint distributions are, more or less, direct extensions of well-known univariate distributions to higher dimensions. The multivariate normal distribution is, among others, often invoked for its mathematical and statistical properties. In the past few years, there has been a growing number of applications giving stronger attention to distributions that can flexibly account for heavier tails \citep{Goodman1973}, skewness \citep{Azzalini1996,Kozubowski2000,Kozubowski2013} and, more in general, to non-elliptical distributions. Rather than assuming a specific parametric distribution, we take a more agnostic approach and we propose to investigate these issues by using directional quantile envelopes (DQEs) for multivariate data \citep{Kong2012}.

The interpretation of a DQE is straightforward. In a bivariate plot, a DQE is represented by a contour line with constant quantile level. Each point on the contour line can be mapped to a fixed percentile of the distribution of the data coordinates' projections onto the real line. The projections are obtained in every possible direction on the circle. The key idea is to divide observations in two groups: those that lie within the contour line (jointly normal) and those outside (jointly abnormal). However, in our specific application investigating the brain sparing effect, the set of all data points classified as jointly abnormal is unsatisfactory as it comprises a clinically heterogeneous mix of infants. In particular, it does not distinguish between infants with symmetric growth and those with asymmetric growth. Rather, we want to focus on infants that, on the one hand, have jointly abnormal measurements, but, on the other, have a large HC-to-BW ratio. This is equivalent to choosing a directional quantile in a particular direction. We propose using allometry to determine such a direction.

Besides the study of \cite{Kong2012}, there are other precedents of applications of methods for multivariate quantiles to anthropometric and growth charts. Some authors \citep{McKeague2011} proposed quantile contours based on Tukey's notion of halfspace depth, while others \citep{Wei2008} considered directional reference intervals built around a central point of the distribution (location parameter). All these studies offer nonparametric approaches to identify jointly abnormal measurements when parametric (normality) assumptions are inappropriate. Therefore, we do not claim any novelty in this regard. In contrast, our focus is on the classification of abnormal ratios. Moreover, we bring forward a connection between DQEs \citep{Kong2012} and allometric modelling which, to our knowledge, does not seem to have been reported before. This in turn provides an operational definition of directional quantile that can be used as cutoff for risk assessment, as demonstrated in the analysis of the VON data.

The rest of the paper is organized as follows. In Section~\ref{sec:2}, we give a general overview of VON and some details about the variables of interest. In Section~\ref{sec:3.1}, we briefly discuss quantiles for univariate data and their limitations when used as individual cutoffs in multivariate problems. In Section~\ref{sec:3.2}, we provide the formal definition of DQE and introduce the relationship between directional quantiles and ratios. In Section~\ref{sec:3.3}, we establish the connection between DQEs and allometry. In Section~\ref{sec:4}, we apply the proposed methods to the VON data to investigate mortality risk in proportionately and disproportionately growth-restricted infants. We show that our principled approach stratifies mortality risk more effectively as compared to the commonly-adopted approach based on isometric scaling. Although our methods focus on anthropometry and allometry, extensions and generalizations to other settings are discussed in Section~\ref{sec:5}.

\section{The Vermont Oxford Network}\label{sec:2}
The Vermont Oxford Network (VON) is a nonprofit, voluntary collaboration of health care professionals `dedicated to improving the quality and safety of medical care for newborn infants and their families through a coordinated program of research, education, and quality improvement projects' (\url{https://public.vtoxford.org/about-us/}). The Network was established in 1988 and comprises over 1200 centers (hospitals) with a neonatal intensive care unit (NICU). The Very Low Birth Weight Database collects information from these centers which account for approximately $90\%$ of all the births occurring at 22-29 weeks of gestation in the United States (US). A number of variables are collected from each center: maternal characteristics (e.g., ethnicity), infant characteristics (e.g., sex, gestational age, birthweight, head circumference, birth defects), and newborn health outcomes (in-hospital mortality and major morbidities). Member hospitals collect the data using uniform definitions through medical record abstraction which are then submitted to VON electronically or through paper forms. Data pass automated checks and are returned for correction if needed.

\begin{table}[ht]
\caption{Sample size ($N$), median birthweight (BW, grams), and median head circumference (HC, centimeters) of infants in the VON dataset.\label{tab:1}}
\centering
\begin{tabular}{lrrrrrr}
  \toprule
  & \multicolumn{3}{c}{\textit{Females}} & \multicolumn{3}{c}{\textit{Males}}\\
  \midrule
  Gestational age & $N$ & BW & HC & $N$ & BW & HC \\
  \midrule
  (21,23] & 5,496 & 550 & 20.5 & 6,036 & 580 & 21.0 \\
  (23,25] & 17,398 & 680 & 22.0 & 19,305 & 730 & 22.5 \\
  (25,27] & 22,448 & 890 & 24.0 & 24,996 & 950 & 24.5 \\
  (27,29] & 28,097 & 1,145 & 26.0 & 31,970 & 1,220 & 26.5 \\
   All ages & 73,439 & 880 & 24.0 & 82,307 & 945 & 24.5 \\
   \bottomrule
\end{tabular}
\end{table}

Our study sample was restricted to inborn, singleton US infants born at 22 to 29 weeks of gestation between 2006 and 2014, with no congenital malformations. GA was determined using obstetrical measures based on prenatal ultrasound (accuracy $\pm7$ days), last menstrual period (accuracy $\pm14$ days), or a neonatologist's estimate based on postnatal physical examinations (accuracy $\pm13$ days for Dubowitz examination). BW was recorded from labor and delivery or, if unavailable, upon admission to the neonatal unit. HC was recorded on the day of birth or the day after. The data underwent some mild cleaning procedures as described elsewhere \citep{Boghossian2016}. In particular, we excluded infants with missing information on vital status (841), unknown gender (30), missing (71) or implausible (744) BW, missing (24,706) or implausible (860) HC, missing hospital length-of-stay (78), or who were hospitalized for longer than one year (565). Overall, about $15\%$ infants were excluded, leaving 155,746 infants for our analysis. Table~\ref{tab:1} gives summary statistics of the sample by sex and four gestational age intervals: (21,23], (23,25], (25,27], and (27,29] weeks. We have divided gestational age in intervals mainly for practical reasons as we want to provide a summary of the results which can be readily used by public health and clinical practitioners. Each gestational age group comprises infants whose mortality risk is comparable within that interval. The sample size of the first interval is noticeably smaller than the other three (Table~\ref{sec:1}), but this interval stands out because it includes infants with the highest risk of mortality \citep{Boghossian2018b}. These data were used in previous publications \citep{Boghossian2016,Boghossian2018} to generate BW- and HC-for-gestational-age percentile charts for clinical use.

\section{Quantile-based risk classification}\label{sec:3}
\subsection{Quantiles of univariate data}\label{sec:3.1}
The quantile function (QF) of a random variable $Y$ with cumulative distribution function (CDF) $F_{Y}(y)$ is defined as
\begin{equation}
\label{eq:1}
Q_{Y}(p)  = \inf \{y \in \mathbb{R}: F_{Y}(y) \geq p\}, \qquad 0 < p < 1.
\end{equation}
Here we assume that $Y$ is absolutely continuous with probability density function $f_{Y}(y) > 0$ over the support of $Y$. Therefore the QF is simply the inverse of the CDF, $Q_{Y}(p) \equiv F^{-1}_{Y}(p)$.

In the presence of covariates, the QF can be extended to conditional distributions. The linear specification of the QR model is \citep{Koenker1978}
\begin{equation}
\label{eq:2}
Q_{Y|\mathbf{X}}(p) = \mathbf{x}\tp \bm\beta(p),
\end{equation}
where $\mathbf{X}$ is a $q$-dimensional vector and $\bm\beta(p)$ is a vector of $q$ coefficients indexed by the quantile level $p$. A generalization of \eqref{eq:2} defines
\begin{equation}
\label{eq:3}
Q_{Y|\mathbf{X}}(p) = h\{\mathbf{x}\tp \bm\beta(p)\},
\end{equation}
where the transformation $h$ can be modelled either parametrically or nonparametrically. Moreover, if $h$ is monotone, then $Q_{h^{-1}(Y)|\mathbf{X}}(p) = \mathbf{x}\tp \bm\beta(p)$, which we call transformation rule \citep{Gilchirst2000} (also known as equivariance to monotone transformations). It is worth mentioning here that quantiles enjoy a number of other properties \citep{Gilchirst2000}, including the reflection rule $Q_{-Y|\mathbf{X}}(p) = -Q_{Y|\mathbf{X}}(1-p)$.

In clinical settings, it is customary to define a cutoff for abnormal measurements. We will use the terms `normal', `subnormal' (below normal), `supranormal' (above normal), and `abnormal' (not normal, either below or above) to classify observations based on arbitrary cutoffs but without giving them any clinical or diagnostic connotation. Cutoffs are often related to specific quantiles of the distribution. For example, infants are classified as SGA if their BW is below the 10th percentile of the BW distribution conditional on gestational age; otherwise, they are termed appropriate for gestational age (AGA). Assuming a model as in \eqref{eq:3}, the cutoff would be determined as
\begin{equation*}
Q_{\mathrm{BW}|\mathrm{GA}}(0.1) = h_{\mathrm{BW}}\{\beta_{0}(0.1) + \beta_{1}(0.1) \cdot \mathrm{GA}\},
\end{equation*}
for some suitable transformation $h_{\mathrm{BW}}$ \citep{Geraci2015,Boghossian2016}. Similarly, infants are said to have a subnormal head size if their HC is below the 10th percentile of the HC distribution conditional on GA, that is
\begin{equation*}
Q_{\mathrm{HC}|\mathrm{GA}}(0.1) = h_{\mathrm{HC}}\{\beta_{0}(0.1) + \beta_{1}(0.1) \cdot \mathrm{GA}\},
\end{equation*}
assuming, as before, that model \eqref{eq:3} holds for some transformation $h_{\mathrm{HC}}$ \citep{Geraci2015,Boghossian2016}. Clearly, these cutoffs need be estimated, either externally using a representative sample from the standard or referent population of interest, or internally from within the same data.

Most of the times, BW and HC are treated separately in statistical analyses. For example, they are analyzed as separate outcomes or as `independent' predictors of postnatal child outcomes. However, there are good reasons why it could be informative to analyze these variables \emph{jointly}. First of all, BW and HC are necessarily related. Larger weights correspond to larger head circumferences, although the younger the baby the larger the head size in relation to the size of body. In other words, the HC-to-BW ratio decreases with age (Figure~\ref{fig:1}). Secondly, health outcomes may differ among infants whose BW and HC are ranked jointly normal or abnormal, as those with symmetric and asymmetric growth. Joint ranking necessitates multivariate approaches.

\subsection{Quantiles of multivariate data}\label{sec:3.2}
Let $\mathbf{Y} = (Y_{1}, \ldots, Y_{K})\tp$ denote a multivariate random vector collecting measurements for $K$ continuous variables (e.g., BW and HC) and let $\mathbf{d} = (d_{1}, \ldots, d_{K})\tp$ be a normalized (with unit norm) direction of dimension $K$. The $p$th directional quantile, in the direction $\mathbf{d}$, is the $p$th quantile of the corresponding projection of the distribution of $\mathbf{Y}$, that is $Q_{\mathbf{d}\tp \mathbf{Y}}(p)$. The supporting half-space determined by $\mathbf{d}$ is $H(\mathbf{d},\xi) = \{\mathbf{y}: \mathbf{d}\tp \mathbf{y} \leq \xi\}$. The $p$th directional quantile envelope (DQE) generated by $Q_{\mathbf{d}\tp \mathbf{Y}}(p)$ is given by the intersection \citep{Kong2012}
\begin{equation}
\label{eq:4}
D(p) = \bigcap_{\mathbf{d}} H\left(\mathbf{d},Q_{\mathbf{d}\tp \mathbf{Y}}(p)\right), \qquad 0 < p \leq 0.5.
\end{equation}

In a bivariate space ($K = 2$), the geometric intuition behind \eqref{eq:4} is as follows. Consider a scatter of points as that on the left plot in Figure~\ref{fig:1}; fix $p$ equal to 0.1; and define a direction on the circle (for example, the west-east or south-north direction). Next, we cumulate data points while moving along the chosen direction and we stop when the cumulative proportion is $10\%$. We demarcate a line which divides the plane into two half-planes, a `lower' half-plane with $10\%$ of the points, and an `upper' half-plane with the remaining $90\%$. If we repeat this process for all the possible directions on the circle, then the intersection of all the demarcation lines defines an oval-shaped contour within which the data points belong to the upper half-planes in \emph{all} directions (an illustration is given in Figure~\ref{fig:2}). These data points represent the set $D(0.1)$. Similarly, the points outside the perimeter, which we denote by $\overline{D(0.1)}$, belong to the lower half-planes in \emph{some} directions.

From the above exemplification it becomes clear that the proportion of points that are in $D(p)$ will be less than $(1-p)$ since these points satisfy $\mathbf{d}\tp \mathbf{y} \geq Q_{\mathbf{d}\tp \mathbf{Y}}(p)$ for all $\mathbf{d} \in \mathbb{R}^{K}$. It also becomes clear that the quantile $p$ in a given direction is equivalent to the quantile $1 - p$ in the opposite direction, e.g., the $10$th directional percentile in the south-north direction is equivalent to the $90$th directional percentile in the north-south direction (hence, $0 < p \leq 0.5$).

Directional quantiles can be easily extended to conditional distributions. If we assume a linear model as in \eqref{eq:2}, we obtain
\begin{equation}
\label{eq:5}
Q_{\mathbf{d}\tp \mathbf{Y}|\mathbf{X}}(p) = \mathbf{x}\tp \bm\beta(p).
\end{equation}
Then the DQE \eqref{eq:4} can be applied to the conditional quantiles in \eqref{eq:5}.

Model \eqref{eq:5} presupposes the additivity of the coordinates of $\mathbf{Y}$ since $\mathbf{d}\tp \mathbf{Y} = d_{1}Y_{1} + d_{2}Y_{2}$. If, instead, a multiplicative relationship is to be studied, then the logarithmic transformation of the measurements is more appropriate \citep{Kong2012}, i.e.
\begin{equation}
\label{eq:6}
Q_{\mathbf{d}\tp \mathbf{Z}|\mathbf{X}}(p) = \mathbf{x}\tp \bm\beta(p),
\end{equation}
where $\mathbf{Z} = (\log Y_{1}, \log Y_{2})\tp$. The $p$th directional quantiles of the coordinates on the log-scale therefore corresponds to the $p$th quantile of the log-ratio of the scaled coordinates, that is
\begin{equation}
\label{eq:7}
d_{1}\log Y_{1} + d_{2} \log Y_{2} = \log\left(\dfrac{Y_{2}^{d_{2}}}{Y_{1}^{-d_{1}}}\right).
\end{equation}

The ratios of the type $R = Y_{2}^{d_{2}}/Y_{1}^{-d_{1}}$ are ubiquitous in public health and clinical applications. One can immediately recognize that $R$ is the BMI index if $Y_{1}$ is height, $Y_{2}$ is weight, and $\mathbf{d} = (-2,1)\tp$. The determination of this particular direction might seem somewhat axiomatic but, far from it, it is the result of statistical observations by Quetelet himself and, some 130 years later, Ancel Keys. However, Quetelet in his treatise \citep{Quetelet1842} had already recognized that the choice of scaling was rather complex:
\begin{quote}
If man increased equally in all dimensions, his weight at different ages would be as the cube of his height. Now, this is not what we really observe. The increase of weight is slower, except during the first year after birth; [...] However, if we compare two individuals [...] we shall find that \emph{the weight of developed persons, of different heights, is nearly as the square of the stature}.
\end{quote}

Allometry studies the geometrical relationships in the human body. In the following section, we discuss the connection between allometry and DQEs. Such a connection provides an operational definition of directional quantile that can be used as cutoff for risk assessment.

\subsection{Bivariate percentiles and allometric analysis}\label{sec:3.3}
Let us consider an allometric model of the type
\begin{equation}\label{eq:8}
Y_{2} = a Y_{1}^{b},
\end{equation}
$a,b > 0$. Equation \eqref{eq:8} implies that the \textit{allometric ratio} $Y_{2}/Y_{1}^b$ is constant and equal to $a$.

In our specific application where $Y_{1}$ is BW and $Y_{2}$ is HC, the scaling exponent $b$ captures the differential growth ratio between the head and the body as a whole. For $b=1$, the variables $Y_{1}$ and $Y_{2}$ are said to be isometric. Model \eqref{eq:8}, if correctly specified, provides a benchmark against which we can classify abnormal HC-to-BW ratios. This is asserted in the following proposition which establishes the connection between \eqref{eq:4} and \eqref{eq:8}.

\begin{proposition}\label{prop:1}
Let $Y_{1}$ and $Y_{2}$ be two continuous and strictly positive random variables, and assume that the allometric model $\mathrm{Y_{2}} = a \mathrm{Y_{1}}^{b}$, $a,b >0$, holds true. Also, define $\mathbf{Z} = (\log Y_{1}, \log Y_{2})\tp$ and assume that the $p$th directional quantile envelope $D(p)$ generated by $Q_{\mathbf{d}\tp \mathbf{Z}}(p)$ is smooth. Then the lines
\begin{align}\label{eq:9}
\nonumber \log Y_{2} &= \log\{Q_{R}(p)\} + b\cdot \log Y_{1}\\
\log Y_{2} &= \log\{Q_{R}(1-p)\} + b\cdot \log Y_{1},
\end{align}
where $R = \dfrac{Y_{2}}{Y_{1}^b}$, are tangent to $D(p)$.
\end{proposition}

\begin{proof}
We only need to prove that the line tangent to $D(p)$ is of the form \eqref{eq:9}. By definition of $D(p)$, the former is given by
\[
Z_{2} = \dfrac{1}{d_{2}}Q_{\mathbf{d}\tp \mathbf{Z}}(p) - \dfrac{d_{1}}{d_{2}}Z_{1}
\]
for any given direction $\mathbf{d}$.

On the log-scale, the allometric equation given in Proposition~\ref{prop:1} relating $Y_{1}$ to $Y_{2}$ can be re-written as
\begin{equation}\label{eq:10}
b \log(Y_{1}) - \log(Y_{2})= -\log a,
\end{equation}
which has the same form of \eqref{eq:7} with $d_{1} = b$ and $d_{2} = -1$. Therefore, for $\mathbf{d} = (b,-1)\tp$, which we call \emph{allometric direction}, the line tangent to $D(p)$ is $\log Y_{2} = -Q_{\mathbf{d}\tp \mathbf{Z}}(p) + b\log Y_{1}$. Now, by \eqref{eq:7} we have that $\mathbf{d}\tp \mathbf{Z} = -\log\left(Y_{2}/Y_{1}^{b}\right) = -\log R$. Since the logarithm is a monotone transformation, we use the transformation rule introduced in Section~\ref{sec:3.1} and obtain $Q_{\mathbf{d}\tp \mathbf{Z}}(p) = Q_{-\log R}(p) = -\log\{Q_{R}(p)\}$. Then the tangent line equation becomes $\log Y_{2} = \log\{Q_{R}(p)\} + b\log Y_{1}$, which corresponds to the first equation given in \eqref{eq:9}. To obtain the second equation in \eqref{eq:9}, it is sufficient to notice that $Q_{-\log R}(p) = -Q_{\log R}(1-p) = -\log\{Q_{R}(1-p)\}$, where the first equality follows from the reflection rule and the second equality follows from the transformation rule.

For a rigorous proof of the geometric properties of $D(p)$, the reader is referred to \cite{Kong2012}.
\end{proof}

Two corollaries to Proposition~\ref{prop:1} follow.
\begin{corollary}
The tangent lines are unique.
\end{corollary}

\begin{corollary}
The tangent half-spaces are the sets of points $\mathbf{Y} = (Y_{1},Y_{2})\tp$ such that $R < Q_{R}(p)$ and $R > Q_{R}(1-p)$.
\end{corollary}

The first corollary follows from the smoothness of $D(p)$. The second corollary is a consequence of the definition of $D(p)$. More importantly, this corollary provides the operational definition of subnormal ($R < Q_{R}(p)$) and supranormal ($R > Q_{R}(1-p)$) ratios corresponding to the allometric direction. The dashed line in the right plot of Figure~\ref{fig:2} gives an illustration of the allometric direction. One may wonder if there is anything special about this direction. The answer lies in the properties of the estimator of $b$. In particular, if $b$ is estimated using MA regression, the directional quantile in the allometric direction is in the same direction as the principal axis of the bivariate normal ellipse fitted to the log-transformed data. It is well-known that this is the direction of the first eigenvector of the variance-covariance matrix of $\mathbf{Z}$. Of course, there is nothing necessarily prescriptive about the normal distribution, so one can explore an alternative estimator for $b$ under a different distribution, if that distribution has a theoretical or empirical relevance, or use a nonparametric estimator. For example, it is common to estimate the slope of MA regression under the assumption of homoscedasticity for the log-additive counterpart of model \eqref{eq:8}. If necessary, this assumption may be relaxed to improve on accuracy and efficiency of the estimates by introducing a variance function of the type ${\var}(Y_2)=\tau^2Y_1^c$, where $\tau > 0$ and $c>0$. Alternatively, using a distribution-free approach, one could estimate $b$ by means of median regression on the log-scale \citep{Geraci2013}, with advantages in terms of robustness to outliers and error distribution, as well as in terms of lossless transformation between scales due to the equivariance property. Moreover, such estimator has a close relationship with the Laplace distribution. However, the median estimator would not be robust to measurement error in the covariates and hence would require the application of methods that are computationally more complex than MA regression \citep{Wei2012,Mao2017}.

In general, there is a stronger motivation for using the allometric direction and this is related to body proportionality in human growth assessment. If we consider studies on infants, the definition of `proportionality' varies from study to study, where body weight is sometimes related to HC, or abdomen circumference, or more commonly to length. Except for some studies (e.g., as those based on the Rohrer index \citep{Olsen2009}), several other implicitly assume isometric scaling by defining ratios of the type $R_{0} = Y_{2}/Y_{1}$ \citep{Kramer1990a,Lin1991,Williams1998,Dashe2000}. However, if the true scaling exponent $b$ is different from one, then the ratio
\[
R_{0} = a Y_{1}^{b-1}
\]
depends on $Y_{1}$. This requires that a definition of abnormal ratio should be based on the conditional quantile $Q_{R_{0}|Y_{1}}(p)$, not on the marginal $Q_{R_{0}}(p)$. The consequence of using the latter would be a misclassification of infants in categories of possibly different risks, which is the source of differential misclassification bias \citep{Heymsfield2007,DPCG2005} we referred to in Section~\ref{sec:1}. We provide empirical evidence that misclassification bias and, in turn, poorer risk classification result in the analysis of the VON data when $b$ is arbitrarily fixed equal to 1 (Section~\ref{sec:4.2}).

Our discussion so far has focused on two variables only, mainly because we are interested in asymmetric growth restriction which is commonly defined using BW and HC only \citep{Lin1991,Bocca2014,Guellec2015}, but also because the VON data do not provide anthropometric variables besides BW and HC. However, it is natural to consider a generalization for $K > 2$. One of the first problems we would encounter is, obviously, visualizing a multivariate DQE in more than, say, three or four dimensions. However, computationally \eqref{eq:4} can be applied for any $K \geq 2$. In contrast, equation \eqref{eq:8} does not seem to have an immediate multivariate counterpart and different approaches can be considered. A popular approach to multivariate allometry is based on principal component analysis \citep{Jolicoeur1963,Corruccini1983} (PCA) due to its geometric properties. In Appendix~\ref{sec:A}, we provide a generalization of Proposition~\ref{prop:1} using PCA to sketch the main idea.

\section{Mortality risk in preterm infants}\label{sec:4}
\subsection{Risk classification based on univariate percentiles}\label{sec:4.1}

As discussed in Section~\ref{sec:1}, SGA infants and infants with subnormal HC are at increased risk of poor health outcomes. We can investigate exposure-outcome associations using appropriate regression models. In particular, if we have both BW and HC as exposures, we can define a categorical variable $u$ with categories `normal BW, normal HC', `subnormal BW, normal HC', `subnormal BW, subnormal HC', and `normal BW, subnormal HC'. In Figure~\ref{fig:1}, these categories correspond to points in quadrants I, II, III, and IV, respectively. Let $W$ be the outcome of interest (death) and $\mathbf{x}$ a vector of covariates associated with $W$ (sex, gestational age, and their interaction). Using the VON data, we fitted the generalized linear model
\[
\expect(W) = \phi^{-1}\left(u, \mathbf{x}\right)
\]
with log-link function $\phi$.

Table~\ref{tab:2} shows the gestational-age-adjusted mortality risk of preterm infants born between 22 and 29 weeks. The baseline is given by infants with normal (i.e., above the 10th percentile for gestational age and sex) BW and HC. The baseline risk is $11\%$ and $14\%$ for girls and boys, respectively. The risk increases by $67\%$ (girls) or $41\%$ (boys) when only HC is subnormal and by about $80\%$ when only BW is subnormal. However, if \emph{both} HC and BW are below their respective cutoffs, then the mortality risk is approximately 2.5 times the baseline risk, meaning that 1 out of 3 of these infants does not survive. The gestational-age-specific mortality risk for girls and boys is given in Appendix~\ref{sec:B} (Tables~\ref{tab:B.1} and \ref{tab:B.2}, respectively). The relative risk for infants with subnormal BW and HC increases with gestational age, although the baseline risk of normal infants is highest at the lowest gestational ages.

\begin{table}[h!]
\caption{The gestational-age-adjusted mortality risk and $95\%$ confidence interval for infants born preterm (22 to 29 weeks) with normal birthweight (BW) and head circumference (HC) are shown in bold font. The other rows show the mortality relative risk (as compared to infants with normal BW and HC) and $95\%$ confidence interval for infants with either one or both anthropometric measurements below the univariate 10th percentile. Estimates are given by sex. The sample size is denoted by $N$.\label{tab:2}}
\centering
\begin{tabular}{llrrrr}
    \toprule
BW & HC & $N$ & Risk & Lower & Upper \\
   \midrule
   \multicolumn{6}{l}{\textit{Females}} \\
   \midrule
Normal & Normal & 63,523 & \textbf{0.11} & \textbf{0.11} & \textbf{0.11} \\
  Normal & $<$ 10th & 2,616 & 1.67 & 1.54 & 1.82 \\
  $<$ 10th & Normal & 3,155 & 1.82 & 1.69 & 1.96 \\
  $<$ 10th & $<$ 10th & 4,145 & 2.57 & 2.44 & 2.71 \\
  \midrule
  \multicolumn{6}{l}{\textit{Males}} \\
  \midrule
  Normal & Normal & 71,070 & \textbf{0.14} & \textbf{0.13} & \textbf{0.14} \\
  Normal & $<$ 10th & 3,079 & 1.41 & 1.31 & 1.52 \\
  $<$ 10th & Normal & 3,606 & 1.81 & 1.70 & 1.92 \\
  $<$ 10th & $<$ 10th & 4,552 & 2.43 & 2.32 & 2.54 \\
   \bottomrule
\end{tabular}
\end{table}

The risk categories in Table~\ref{tab:2} are defined using separate rankings for BW and HC. This is the approach taken by some authors \citep{Lin1991,Guellec2015}. However, this approach presents a difficulty. While the specific cutoff values (e.g., 10th percentile) may be relevant for BW and HC taken individually, nothing can be said about the joint ranking of the measurements. We have already noticed that the observation marked by a thick cross in the left plot of Figure~\ref{fig:1} presents normal (for gestational age) weight (489 g) and HC (30 cm) and thus falls in the baseline group. Upon closer inspection, this infant has a rather extreme HC-to-BW ratio (right plot in Figure~\ref{fig:1}), a possible consequence of the brain-sparing effect.

In the next section, we examine the categories in Table~\ref{tab:2} in more detail to see whether the abnormality of the HC-to-BW ratio represents an additional risk factor. In particular, the categories `normal BW, normal HC' and `subnormal BW, subnormal HC' are of clinical interest since they represent the baseline and the highest mortality risk categories, respectively. The category `subnormal BW, normal HC' comprises infants with a BW-HC imbalance. Some studies use this imbalance to define asymmetric growth restriction \citep{Guellec2015}. However, as we previously argued, this definition is limited as it does not take into account the proportionality between BW and HC.

\subsection{Risk classification based on the allometric direction}\label{sec:4.2}

\begin{figure}[h!]
\includegraphics[scale = .40]{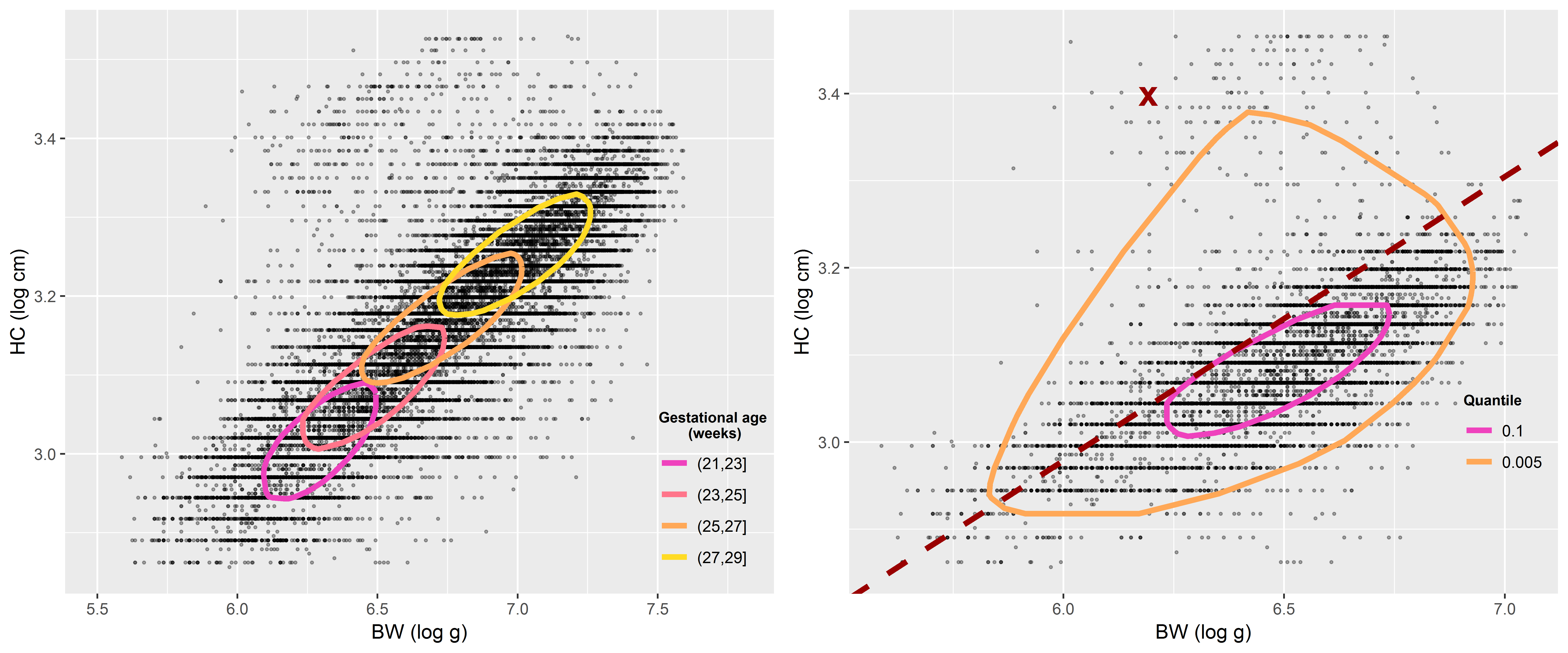}
\caption{Left: gestational-age-specific directional quantile envelopes (DQEs) of birthweight (BW) and head circumference (HC) for all girls at level $p = 0.1$. Right: DQEs of BW and HC at levels $p = 0.1$ and $p = 0.005$ for girls born at $(23,25]$ weeks of gestation. The dashed line marks the 90th percentile of the HC-to-BW ratio in the allometric direction.\label{fig:2}}
\end{figure}

We first explore the bivariate distribution of BW and HC using DQEs conditional on gestational age. We use DQEs at level $p = 0.1$ since the directional quantile $Q_{\mathbf{d}\tp \mathbf{Z}}(0.1)$ coincides with the univariate 10th percentile $Q_{\log \mathrm{BW}}(0.1)$ for $\mathbf{d} = (1,0)\tp$ and $Q_{\log \mathrm{HC}}(0.1)$ for $\mathbf{d} = (0,1)\tp$ which are commonly used as clinical cutoffs.

We can take advantage of model \eqref{eq:6} to estimate the DQE of $\mathbf{Z} = (\log \mathrm{BW}, \log \mathrm{HC})\tp$ at level $p = 0.1$ conditional on gestational age, i.e.
\[
Q_{\mathbf{d}\tp \mathbf{Z}|\mathbf{X}}(0.1) = \beta_{0}(0.1) + \beta_{1}(0.1) x_{1} + \beta_{2}(0.1) x_{2} + \beta_{3}(0.1) x_{3},
\]
where the reference is $(21,23]$ weeks and $x_{j}$, $j = 1, 2, 3$, are dummy variables for the other gestational age intervals. The left plot of Figure~\ref{fig:2} shows the estimated DQEs for girls in each age interval. The points that lie outside a particular contour line (i.e., the set $\overline{D(0.1)}$) represent BW and HC measurements that are \emph{jointly abnormal} as compared to those that fall in $D(0.1)$, conditional on a specific gestational age group. We should note that, while informative, the classification based on the DQE is broad since the `abnormal' labelling of measurements outside the DQE applies to disparate groups: infants with subnormal BW and HC, infants with supranormal BW and HC, and those with asymmetric BW and HC. We may want to focus on a particular group of measurements that are jointly abnormal, like those for which $R > Q_{R}(0.9)$ (supranormal HC-to-BW ratio) that are suggestive of the brain-sparing effect.

We then fitted model~\eqref{eq:8} for BW and HC on the log-scale using standardized major axis (MA) regression as implemented in the \texttt{R} package \texttt{smatr} \citep{Warton2012}. The reason for this choice lies in the likely presence of measurement error in both variables \citep{Warton2006}. Table~\ref{tab:3} shows estimates of the coefficients and standard errors for all infants as well as by sex and by gestational age. Overall, the estimated coefficient $\hat{b} \approx 0.32$ hints at the cubic relationship between length and volume, although the test of the null hypothesis $H_{0}: b = 1/3$ gave a $p$-value less than 0.001. The estimated scaling exponent was almost identical for boys and girls ($p = 0.841$), but changed significantly ($p < 0.001$) across gestational ages. The rightmost column of Table~\ref{tab:3} shows the $90$th percentile of the estimated allometric ratio $R = \mathrm{HC}/\mathrm{BW}^{\hat{b}}$. An illustration for girls born at $(23,25]$ weeks is shown in the right plot of Figure~\ref{fig:2}. The dashed line marks the $90$th percentile of the HC-to-BW ratio in the allometric direction (note that this is tangent to the DQE at level $p = 0.1$). Points to its left have a relatively large value of HC as compared to that of BW. Among these, marked by a thick cross, we find the girl born at 162 days of gestation that was featured in Figure~\ref{fig:1}. Without question, this observation seems to be rather extreme even at $p = 0.005$.

\begin{table}[t!]
\caption{Estimates (standard errors) of the coefficients of the allometric model for BW and HC, along with the $p$-values of the tests on equality of the slopes. The $90$th percentile of the allometric ratio is reported in the last column.\label{tab:3}}
\centering
\begin{tabular}{lrrr}
  \toprule
    & \multicolumn{3}{c}{Parameter}\\
    & \multicolumn{1}{c}{$\log_{10} a$} & \multicolumn{1}{c}{$b$} & \multicolumn{1}{c}{$Q_{R}(0.9)$}\\
  \midrule
  \textit{All} & 0.4488 (0.0011) & 0.3166 (0.0004) & 2.9647 (0.0006) \\
  \midrule
  \multicolumn{3}{l}{\textit{Sex ($p$-value $0.841$)}} \\
  \midrule
  Females & 0.4488 (0.0017) & 0.3166 (0.0006) & 2.9641 (0.0008) \\
  Males & 0.4497 (0.0016) & 0.3164 (0.0005) & 2.9704 (0.0009) \\
  \midrule
  \multicolumn{3}{l}{\textit{Gestational age ($p$-value $< 0.001$)}}\\
  \midrule
  (21,23] & 0.1816 (0.0087) & 0.4125 (0.0032) & 1.6123 (0.0018) \\
  (23,25] & 0.4212 (0.0034) & 0.3260 (0.0012) & 2.7798 (0.0013) \\
  (25,27] & 0.5102 (0.0027) & 0.2963 (0.0009) & 3.4026 (0.0011) \\
  (27,29] & 0.5413 (0.0024) & 0.2865 (0.0008) & 3.6553 (0.0010) \\
   \bottomrule
\end{tabular}
\end{table}

Using the results in Table~\ref{tab:3}, we can calculate the allometric ratio $R = \mathrm{HC}/\mathrm{BW}^{\hat{b}_{j}}$ for each infant, where $\hat{b}_{j}$ is the estimated coefficient for gestational age group $j$, $j = 1, \ldots, 4$, and classify these ratios based on the gestational-age-specific cutoffs $\hat{Q}_{R}(0.9)$. Table~\ref{tab:4} compares gestational-age-adjusted mortality in infants with normal ($R \leq Q_{R}(0.9)$) and supranormal ($R > Q_{R}(0.9)$) ratios. In these two groups, the baseline mortality risk is comparable for infants with normal BW and HC: between 0.11 and 0.13 for females, and 0.14 for males. However, the mortality risk for infants with small BW and HC is about three times the baseline risk in those with supranormal ratios, but about twice the baseline risk in those with a normal ratio. As a consequence, absolute risks too differ greatly. For example, \emph{disproportionately} small boys have an absolute risk of $3.17 \times 0.14 = 0.44$ while \emph{proportionately} small boys have an absolute risk of $2.16 \times 0.14 = 0.30$. Note that in Table~\ref{tab:4} the category of infants with normal BW and subnormal HC has been omitted since, by definition, there are no infants with supranormal ratios and subnormal HC.

In summary, preterm infants that have small BW and HC (below their respective 10th percentiles) are at high risk of mortality, with a relative risk of around 2.5 as compared to infants with normal BW and HC (Table~\ref{tab:2}). However, this relative risk is an `average' of a lower risk in proportionately small infants and a higher risk in disproportionately small infants (Table~\ref{tab:4}). Hence, the classification based on the allometric direction tangent to the DQE identifies the group of disproportionately growth-restricted infants as those with considerably high risk among small infants.

As we repeatedly mentioned, if the correct scaling exponent is not applied, then misclassification bias may result and distort risk assessment. In Appendix~\ref{sec:B} (Table~\ref{tab:B.3}), we report estimates of mortality risk for infants with normal and supranormal ratios, where the ratio is calculated using the isometric scaling, i.e. $R_{0} = \mathrm{HC}/\mathrm{BW}$. It is apparent that the mortality risk for infants with small BW and HC is not dissimilar from the baseline risk in those with supranormal ratios, while the magnitude of the estimated relative risk in those with normal ratios is noticeably smaller than that obtained in Table~\ref{tab:4}. Moreover, the classification based on the isometric scaling leads to relative risk estimates less than (although not significantly different from) 1 in one of the female groups, which is contrary to the well-established notion that SGA infants are at higher mortality risk.

\begin{table}[ht]
\caption{The gestational-age-adjusted mortality risk and $95\%$ confidence interval for infants born preterm (22 to 29 weeks) with normal birthweight (BW) and head circumference (HC) are shown in bold font. The other rows show the mortality relative risk (as compared to infants with normal BW and HC) and $95\%$ confidence interval for infants with either one or both anthropometric measurements below the univariate 10th percentile. Estimates are given by sex, separately for infants whose HC-to-BW ratio in the allometric direction is below the 90th percentile (normal) or above it (supranormal). The sample size is denoted by $N$.\label{tab:4}}
\centering
\begin{tabular}{llrrrrrrrr}
  \toprule
  \multicolumn{1}{l}{BW} & \multicolumn{1}{l}{HC} & \multicolumn{1}{c}{$N$} & \multicolumn{1}{c}{Risk} & \multicolumn{1}{c}{Lower} & \multicolumn{1}{c}{Upper} & \multicolumn{1}{c}{$N$} & \multicolumn{1}{c}{Risk} & \multicolumn{1}{c}{Lower} & \multicolumn{1}{c}{Upper}\\
  \midrule
  \multicolumn{2}{l}{\textit{Females}} & \multicolumn{4}{c}{\textit{Normal HC-to-BW ratio}} & \multicolumn{4}{c}{\textit{Supranormal HC-to-BW ratio}}\\
  \midrule
  Normal & Normal & 59,033 & \textbf{0.11} & \textbf{0.11} & \textbf{0.11} & 4,490 & \textbf{0.13} & \textbf{0.12} & \textbf{0.14} \\
  $<$ 10th & Normal & 1,425 & 1.70 & 1.52 & 1.90 & 1,730 & 1.62 & 1.44 & 1.82 \\
  $<$ 10th & $<$ 10th & 3,284 & 2.30 & 2.16 & 2.45 & 861 & 3.14 & 2.82 & 3.50 \\
  \midrule
  \multicolumn{2}{l}{\textit{Males}} & \multicolumn{4}{c}{\textit{Normal HC-to-BW ratio}} & \multicolumn{4}{c}{\textit{Supranormal HC-to-BW ratio}}\\
  \midrule
  Normal & Normal & 65,984 & \textbf{0.14} & \textbf{0.13} & \textbf{0.14} & 5,086 & \textbf{0.14} & \textbf{0.13} & \textbf{0.15} \\
  $<$ 10th & Normal & 1,425 & 1.77 & 1.61 & 1.95 & 2,181 & 1.81 & 1.64 & 2.00 \\
  $<$ 10th & $<$ 10th & 3,384 & 2.16 & 2.04 & 2.28 & 1,168 & 3.17 & 2.88 & 3.48 \\
  \bottomrule
\end{tabular}
\end{table}

\subsection{Asymmetric growth and hypertension}

We investigated maternal hypertension, which has been previously found to be a determinant of the HC-to-BW ratio in its severe, pregnancy-induced form \citep{Kramer1990a}. While information on hypertension is available in the VON data, unfortunately this variable has two limitations: it includes both chronic and pregnancy-induced hypertension (PIH), and is missing for about $21\%$ (though mostly in early years of data collection). Yet, some interesting observations can be made.

Hypertension is known to increase the likelihood of growth restriction. This is apparent from Figure~\ref{fig:3} which shows estimated DQEs conditional on hypertension status. Its relationship with mortality risk is, however, controversial. Some studies suggested that PIH increases the risk of fetal, perinatal, and early neonatal mortality \citep{Jain1997}, while other studies found the opposite \citep{Chen2006}. In our data, the prevalence of hypertension (chronic and gestational) is overall about $28\%$, while infant mortality rates are approximately $0.15\%$ and $0.11\%$ in, respectively, normotensive and hypertensive mothers, suggestive of a `protective' effect of hypertension. However, the rate of hypertension is $25\%$ among mothers of babies with normal HC-to-BW ratio, but $54\%$ in mothers of disproportionately small babies. In other words, there is a strong, positive association between hypertension and supranormal ratios ($\chi^2$ test's $p$-value $< 0.001$). As shown in Table~\ref{tab:5}, the mortality risk in infants with normal BW and HC is lower if born to hypertensive mothers as compared to normotensive mothers, regardless of their HC-to-BW ratio. However, disproportionately small infants born to hypertensive mothers have an absolute risk of $6.21 \times 0.07 = 0.43$ which, compared to an absolute risk of $2.40 \times 0.17 = 0.41$ in their peers born to normotensive mothers, gives a rather different picture of the association between mortality and hypertension. It has been speculated that, in preterm infants born to hypertensive mothers, maternal hypertension is less damaging for fetal development than other causes of growth restriction \citep{McBride2017}. Our results do not exclude this hypothesis, but they also point to an interaction between hypertension (presumably its severe forms) and asymmetric growth restriction.

\begin{figure}[h!]
\includegraphics[scale = .5]{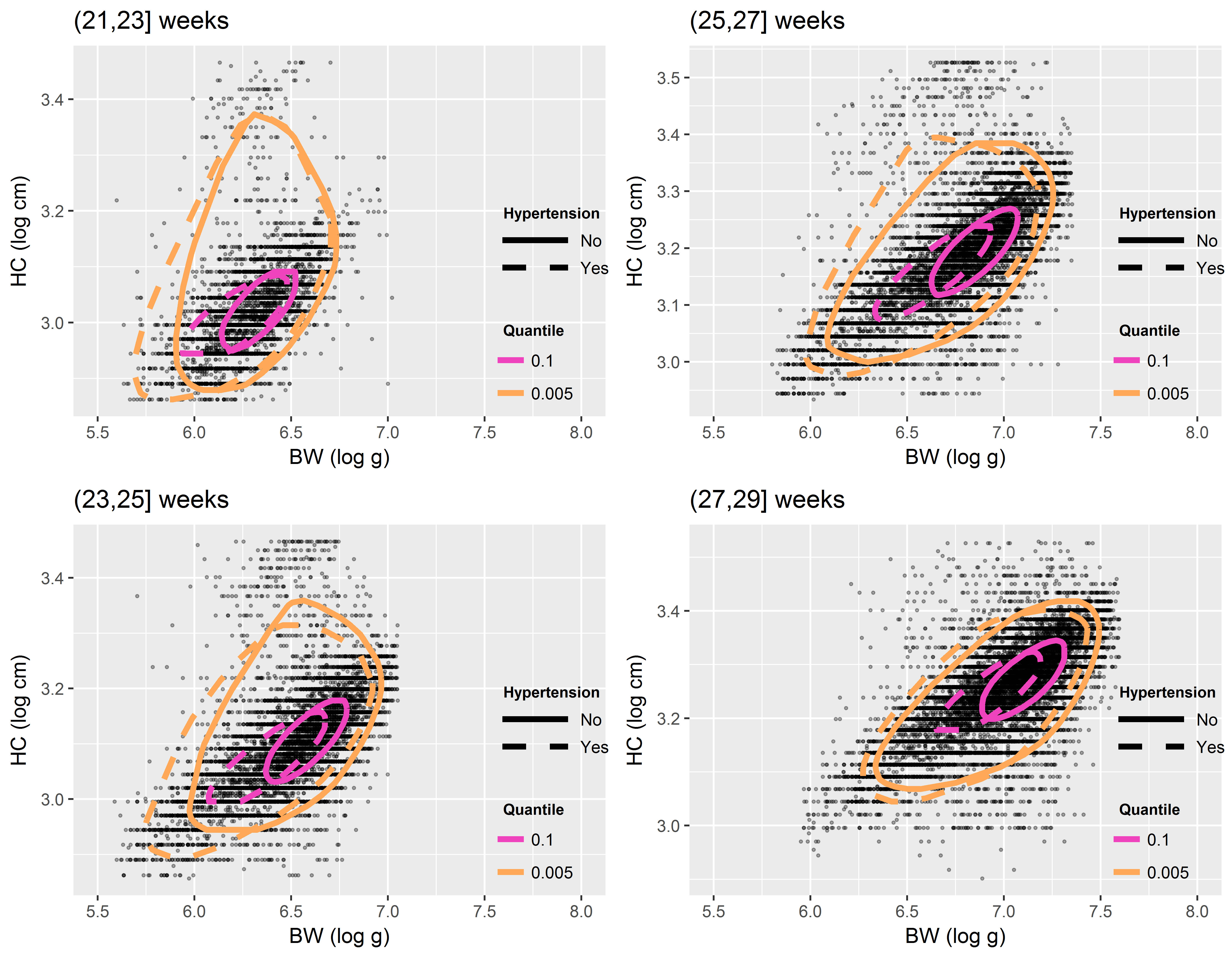}
\caption{Directional quantile envelopes (DQEs) of birthweight (BW) and head circumference (HC) at levels $p = 0.1$ and $p = 0.005$ for infants born to normotensive (solid lines) and hypertensive mothers at different gestational ages.\label{fig:3}}
\end{figure}

\begin{table}[ht]
\caption{The gestational-age-adjusted mortality risk and $95\%$ confidence interval for infants born preterm (22 to 29 weeks) with normal birthweight (BW) and head circumference (HC) are shown in bold font. The other rows show the mortality relative risk (as compared to infants with normal BW and HC) and $95\%$ confidence interval for infants with either one or both anthropometric measurements below the univariate 10th percentile. Estimates are given for groups divided by normal and supranormal HC-to-BW ratio in the allometric direction, separately for infants born to normotensive and hypertensive mothers. The sample size is denoted by $N$.\label{tab:5}}
\centering
\begin{tabular}{llrrrrrrrr}
  \toprule
  \multicolumn{1}{l}{BW} & \multicolumn{1}{l}{HC} & \multicolumn{1}{c}{$N$} & \multicolumn{1}{c}{Risk} & \multicolumn{1}{c}{Lower} & \multicolumn{1}{c}{Upper} & \multicolumn{1}{c}{$N$} & \multicolumn{1}{c}{Risk} & \multicolumn{1}{c}{Lower} & \multicolumn{1}{c}{Upper}\\
  \midrule
  \multicolumn{2}{l}{\textit{Normal HC-to-BW ratio}} & \multicolumn{4}{c}{\textit{Normotensive}} & \multicolumn{4}{c}{\textit{Hypertensive}}\\
  \midrule
Normal & Normal & 76,700 & \textbf{0.13} & \textbf{0.13} & \textbf{0.13} & 22,136 & \textbf{0.07} & \textbf{0.07} & \textbf{0.08} \\
  $<$ 10th & Normal & 666 & 2.65 & 2.39 & 2.94 & 1,631 & 2.05 & 1.81 & 2.32 \\
  $<$ 10th & $<$ 10th & 1,972 & 2.43 & 2.28 & 2.60 & 3,479 & 3.30 & 3.06 & 3.56 \\
  \midrule
  \multicolumn{2}{l}{\textit{Supranormal HC-to-BW ratio}} & \multicolumn{4}{c}{\textit{Normotensive}} & \multicolumn{4}{c}{\textit{Hypertensive}}\\
  \midrule
  Normal & Normal & 3,993 & \textbf{0.17} & \textbf{0.16} & \textbf{0.19} & 3,353 & \textbf{0.07} & \textbf{0.06} & \textbf{0.08} \\
  $<$ 10th & Normal & 951 & 1.80 & 1.61 & 2.03 & 2,080 & 2.75 & 2.36 & 3.22 \\
  $<$ 10th & $<$ 10th & 579 & 2.40 & 2.14 & 2.70 & 1,041 & 6.21 & 5.37 & 7.19 \\
  \bottomrule
\end{tabular}
\end{table}

\clearpage

\section{Discussion}\label{sec:5}
In this study, we proposed an approach to risk classification of abnormal ratios based on the allometric direction, which is intimately connected to directional quantiles \citep{Kong2012}, and we applied these methods to data on birthweight and head circumference in a large cohort of preterm infants. Our analysis suggests that small preterm infants with large HC-to-BW ratio are at increased mortality risk as compared not only to AGA infants, but also to proportionately growth-restricted preterm infants. There is evidence in the literature that asymmetric growth restriction increases the likelihood of adverse outcomes. One study concluded that ``the prognosis of SGA infants with asymmetric growth [defined by the ratio of HC to abdominal circumference] is poorer than that of symmetrically grown infants and much worse than that of AGA infants'' \citep{Dashe2000}. Therefore, our findings are consistent with the literature.

Following our approach, disproportionately growth-restricted infants can be identified as those whose BW and HC lie outside a given quantile contour along the allometric direction. It is straightforward to carry out the identification task. All is needed for classification is the scaling exponent $b$ and the cutoff $Q_{R}(0.9)$. The calculation of the DQE is not necessary. Suppose for instance that an infant is born at 23 weeks and that her BW is 390 g and HC is 19.5 cm. These measurements already put her in a high-risk category since both measurements are below their respective 10th percentile (univariate) thresholds for gestational age and sex \citep{Boghossian2016,Boghossian2018}. Using gestational-age specific estimates from Table~\ref{tab:3}, the allometric HC-to-BW ratio is easily found as $19.5/390^{0.4125} = 1.6643$ which is greater than the cutoff $\hat{Q}_{R}(0.9) = 1.6123$. The statistical justification of this procedure comes from Proposition~\ref{prop:1}, which asserts that the allometric direction is tangent to the DQE and thus guarantees that all the points such that $R > Q_{R}(0.9)$ are a subset of $\overline{D(0.1)}$ (jointly abnormal). Our study offers not only a statistically principled approach to risk classification, but also large-sample estimated cutoffs that can be immediately used by practitioners together with previously published anthropometric charts for BW and HC using the same data \citep{Boghossian2016,Boghossian2018}.

There are some limitations in our data. Firstly, gestational age is subject to measurement error and the accuracy depends on the method of estimation as well as when the measurement is made. We believe that our proposed approach can be extended to account for this error when estimating directional quantiles and allometric directions. Another possible limitation in our analysis is the omission of covariates that might explain different allometric relationships. We partly made up for this deficiency by investigating hypertension, and found that the mortality risk is associated with an interaction between asymmetric growth restriction and hypertension. In a separate analysis (results not shown), we investigated the allometric model \eqref{eq:8} conditional on ethnicity but no meaningful differences were found. Unfortunately, other potentially relevant covariates such as maternal age or parental weight and height, which are known to be associated with birthweight \citep{Griffiths2007,Geraci2016}, are not available in the VON data. There are, however, a number of in-hospital morbidity outcomes which we will explore in a separate study. Furthermore, it may be relevant to extend the analysis to infants with subnormal HC-to-BW ratios. Suboptimal head size at birth is known to be a risk factor for poor neurodevelopmental outcomes if it persists after birth \citep{Hack1991,Kuban2009}. This type of investigation would require follow-up information which is not available to us at this time.

The proposed approach can be extended to applications where models equivalent to \eqref{eq:8} are expected to hold and the interest lies in the classification of ratios. These obviously include applications with other anthropometric ratios, such as the waist-to-hip ratio, which has been proposed as a predictor of newborn size \citep{Brown1996}, the waist-to-height ratio, which has been found associated with prehypertension \citep{Djeric2017}, and the subscapular-to-triceps skinfolds, fat--mass-to-weight, and fat--mass-to-fat--free mass ratio indices which are used as risk factors in stunted populations \citep{Judd2008}. As we move out of anthropometry, we find several other potential applications, especially in the biomedical sciences. For example, insulin and C-peptide are strongly correlated, and their ratio is a biomarker used to discriminate between insulinoma, a tumor of the pancreas, and injection of excessive insulin (surreptitious or inadvertent) \citep{Lebowitz1993}. The ejection fraction is a measure of the ratio between the blood pumped in and out of the left ventricle. A low input volume indicates atrial insufficiency, a low output volume ventricular insufficiency, and a low ratio might be associated to heart failure in several respects (e.g., decline of the contractile function). The list goes on with the albumin-to-globulin ratio (liver insufficiency, immunodeficiency, auto-immunity, infection, cancer); calcium-to-albumin ratio (malnutrition); albumin-to-creatinine ratio (kidney disease); cortisol-to-cortisone ratio (cardiovascular risk); calcitriol-to-calcifediol ratio (renal efficiency) \citep{Rotondi2018}; oxigen extration ratio (haemodialysis efficacy); metabolite ratios (type 2 diabetes) \citep{Molnos2018}; and stable isotope ratios (diet) \citep{OBrien2015}.

\section*{Acknowledgements}
Marco Geraci was funded by an ASPIRE grant from the Office of the Vice President for Research at the University of South Carolina and by the National Institutes of Health -- National Institute of Child Health and Human Development (Grant Number: 1R03HD084807-01A1). The authors wish to thank: Linglong Kong for providing the R code to perform quantile estimation with multivariate data as described in \cite{Kong2012}; Erika Edwards for her support and participation in discussions of earlier drafts of the manuscript; two anonymous reviewers for their valuable comments on an earlier draft of the manuscript. The University of Vermont's committee on human research determined that this study using VON's de-identified research repository was not human subjects research.

\clearpage

\appendix
\section{Appendix -- Multivariate percentiles and allometric analysis}
\label{sec:A}

Let us consider a $k$-dimensional vector $\mathbf{Y}$, $K \geq 2$, and the element-wise log-transformed vector $\mathbf{Z} = (\log Y_{1}, \log Y_{2}, \ldots, \log Y_{K})\tp$. Also, let $\bm\Sigma$ be the positive-definite variance-covariance matrix of $\mathbf{Z}$. The goal of PCA applied to $\mathbf{Z}$ is to obtain the decomposition $\bm\Sigma = \bm\Delta\bm\Lambda\bm\Delta\tp$. (As noted elsewhere \citep{Jolicoeur1963}, the log-transformation removes measurement scale differences and therefore makes the use of the correlation matrix unnecessary. However, some authors \citep{Somers1986} advocate the use of the correlation matrix to separate size and shape variation.) The columns of $\bm\Delta$, say $\bm\delta_{k}$, $k = 1, \ldots, K$, are the unit-norm eigenvectors of the decomposition, while the diagonal matrix $\bm\Lambda$ has diagonal elements $\lambda_{k}$, $k = 1, \ldots, K$, and are given by the corresponding eigenvalues. The PCA scores are then obtained as $\tilde{\mathbf{Z}} = \mathbf{Z}\bm\Delta$.

The first component $\bm\delta_{1}$ (major axis) is the direction with maximal variance and its equation can expressed as \citep{Jolicoeur1963}
\begin{align*}
\frac{1}{\cos \theta_{1}}\left(Z_{1} - \expect(Z_{1})\right) &= \ldots = \frac{1}{\cos \theta_{i}}\left(Z_{i} - \expect(Z_{i})\right) = \frac{1}{\cos \theta_{j}}\left(Z_{j} - \expect(Z_{j})\right) = \ldots\\ &= \frac{1}{\cos \theta_{K}}\left(Z_{K} - \expect(Z_{K})\right),
\end{align*}
where $\theta_{i}$ is the angle made by the first principal component with the coordinate axis of $Z_{i}$. Thus
\[
\left(\frac{Y_{1}}{g_{1}}\right)^{1/\cos \theta_{1}} = \ldots = \left(\frac{Y_{i}}{g_{i}}\right)^{1/\cos \theta_{i}} = \left(\frac{Y_{j}}{g_{j}}\right)^{1/\cos \theta_{j}} = \ldots = \left(\frac{Y_{K}}{g_{K}}\right)^{1/\cos \theta_{K}},
\]
where $g_{i}$ is the geometric mean of $Y_{i}$. It follows that
\begin{equation}\label{eq:11}
Y_{i} = a_{ij} Y_{j}^{b_{ij}}, \qquad i\neq j,\quad j = 1, \ldots, K,
\end{equation}
where $a_{ij} = g_{i}/g_{j}^{b_{ij}}$ and $b_{ij} = \cos \theta_{i}/\cos \theta_{j}$. That is, the ratio of any pair of PCA loadings, say $i$ and $j$, from the first component $\bm\delta_{1}$ approximates the slope of the bivariate MA regression for the variables $Z_{i}$ and $Z_{j}$ \citep{Corruccini1983}. Using these results, we offer the following proposition.

\begin{proposition}\label{prop:2}
Let $\mathbf{Y} = (Y_{1}, \ldots, Y_{K})\tp$ be a multivariate random vector collecting $K$ continuous and strictly positive random variables, and assume that the PCA allometric model \eqref{eq:11} holds true for any pair of variables $i,j$, with $i \neq j$. Also, define $\mathbf{Z} = (\log Y_{1}, \ldots, \log Y_{K})\tp$ and assume that the $p$th directional quantile envelope $D(p)$ generated by $Q_{\mathbf{d}\tp \mathbf{Z}}(p)$ is smooth. Then the hyperplanes
\begin{align}\label{eq:12}
\nonumber \log Y_{i} &= \frac{1}{K-1}\log\{Q_{R_{i}}(p)\} + \sum_{j \neq i} \frac{b_{ij}}{K-1}\log Y_{j}\\
\log Y_{i} &= \frac{1}{K-1}\log\{Q_{R_{i}}(1-p)\} + \sum_{j \neq i} \frac{b_{ij}}{K-1}\log Y_{j},
\end{align}
where $R_{i} = \dfrac{Y_{i}^{K-1}}{g_{-i}}$ and $g_{-i} = \prod_{j\neq i} Y_{j}^{b_{ij}}$, are tangent to $D(p)$.
\end{proposition}

\begin{proof}
In its implicit form, the hyperplane tangent to $D(p)$ is by definition given by
\[
\sum_{j=1}^{K} d_{j}Z_{j} = Q_{\mathbf{d}\tp \mathbf{Z}}(p),
\]
for any given direction $\mathbf{d} = (d_{1}, \ldots, d_{K})\tp$. We also note that
\begin{equation}\label{eq:13}
\mathbf{d}\tp \mathbf{Z} =  \sum_{k=1}^{K} d_{k}\log Y_{k} = \log\left(\frac{Y_{i}^{d_{i}}}{\prod_{j \neq i}Y_{j}^{-d_{j}}}\right).
\end{equation}

On the log-scale, the $K - 1$ allometric equations given in Proposition~\ref{prop:2} relating a given $Y_{i}$ to the other variables $Y_{j}$, $j \neq i$, can be re-written as
\begin{align*}
b_{i1} \log(Y_{1}) - \log(Y_{i}) &= -\log a_{i1}\\
& \vdots \\
b_{i(i-1)} \log(Y_{i-1}) - \log(Y_{i}) &= -\log a_{i(i-1)}\\
b_{i(i+1)} \log(Y_{i+1}) - \log(Y_{i}) &= -\log a_{i(i+1)}\\
& \vdots \\
b_{iK} \log(Y_{K}) - \log(Y_{i}) &= -\log a_{iK}.
\end{align*}
If we take the sum of the terms on the left- and right-hand sides, respectively, we obtain
\begin{equation}\label{eq:14}
\log\left(\frac{Y_{i}^{-(K - 1)}}{\prod_{j \neq i}Y_{j}^{-b_{ij}}}\right) = -\sum_{j \neq i} \log a_{ij},
\end{equation}
which has the same form of \eqref{eq:13} with $d_{j} = b_{ij}$, for all $j \neq i$, and $d_{i} = -(K - 1)$. Therefore, the hyperplane tangent to $D(p)$ in the allometric direction, defined as $\mathbf{d} = (b_{i1}, \ldots, b_{i(i-1)}, -K+1, b_{i(i+1)}, \ldots, b_{iK})\tp$, is $\log Y_{i} = -\frac{1}{K-1}Q_{\mathbf{d}\tp \mathbf{Z}}(p) + \sum_{j \neq i} \frac{b_{ij}}{K-1}\log Y_{j}$. By \eqref{eq:13} we have that
\[
\mathbf{d}\tp \mathbf{Z} = -\log\left(\frac{Y_{i}^{K-1}}{\prod_{j \neq i}Y_{j}^{b_{ij}}}\right) = -\log R_{i}.
\]
Similarly to before, we find $Q_{\mathbf{d}\tp \mathbf{Z}}(p) = Q_{-\log R_{i}}(p) = -\log\{Q_{R_{i}}(p)\}$. Then the tangent hyperplane becomes $\log Y_{i} = \frac{1}{K-1}\log\{Q_{R_{i}}(p)\} + \sum_{j \neq i} \frac{b_{ij}}{K-1}\log Y_{j}$, which corresponds to the first equation given in \eqref{eq:12}. The second equation in \eqref{eq:12} is found with similar arguments as those used in Proposition~\ref{prop:1}.
\end{proof}

Model \eqref{eq:11} is rather flexible as it allows for $K(K-1)/2$ distinct slopes $b_{ij}$ to be estimated. If there is evidence or theoretical justification of equality of slopes, some restrictions can be imposed, for example, by assuming homogeneity of the $b_{ij}$'s for some $j$'s or using an even more parsimonious model of the kind $Y_i = a_{i} \left(\prod_{m_{i}} Y_{m_{i}}\right)^{b_{i}}$, $i = 1, \ldots, K$, where $m_{i}$ indexes a subset of $K - 1$ variables $Y_{j}$, $j \neq i$.

\clearpage

\section{Appendix -- Additional tables}
\label{sec:B}

\begin{table}[h!]
\caption{The gestational-age-specific mortality risk and $95\%$ confidence interval for girls born preterm (22 to 29 weeks) with normal birthweight (BW) and head circumference (HC) are shown in bold font. The other rows show the mortality relative risk (as compared to infants with normal BW and HC) and $95\%$ confidence interval for infants with either one or both anthropometric measurements below the univariate 10th percentile. The sample size is denoted by $N$.\label{tab:B.1}}
\centering
\begin{tabular}{llrrrr}
  \toprule
\multicolumn{1}{l}{BW} & \multicolumn{1}{l}{HC} & \multicolumn{1}{c}{$N$} & \multicolumn{1}{c}{Risk} & \multicolumn{1}{c}{Lower} & \multicolumn{1}{c}{Upper} \\
  \midrule
  \multicolumn{6}{l}{$(21,23]$ weeks}\\
  \midrule
  Normal & Normal & 4,655 & \textbf{0.52} & \textbf{0.51} & \textbf{0.54} \\
  Normal & $<$ 10th & 322 & 1.23 & 1.13 & 1.34 \\
  $<$ 10th & Normal & 306 & 1.36 & 1.26 & 1.47 \\
  $<$ 10th & $<$ 10th & 213 & 1.50 & 1.39 & 1.61 \\
  \midrule
  \multicolumn{6}{l}{$(23,25]$ weeks}\\
  \midrule
  Normal & Normal & 14,979 & \textbf{0.19} & \textbf{0.19} & \textbf{0.20} \\
  Normal & $<$ 10th & 650 & 1.45 & 1.27 & 1.65 \\
  $<$ 10th & Normal & 736 & 1.85 & 1.67 & 2.05 \\
  $<$ 10th & $<$ 10th & 1,033 & 2.69 & 2.51 & 2.88 \\
  \midrule
  \multicolumn{6}{l}{$(25,27]$ weeks}\\
  \midrule
  Normal & Normal & 19,571 & \textbf{0.06} & \textbf{0.06} & \textbf{0.07} \\
  Normal & $<$ 10th & 649 & 1.84 & 1.47 & 2.30 \\
  $<$ 10th & Normal & 860 & 2.12 & 1.77 & 2.54 \\
  $<$ 10th & $<$ 10th & 1,368 & 4.19 & 3.77 & 4.65 \\
  \midrule
  \multicolumn{6}{l}{$(27,29]$ weeks}\\
  \midrule
  Normal & Normal & 24,318 & \textbf{0.02} & \textbf{0.02} & \textbf{0.03} \\
  Normal & $<$ 10th & 995 & 1.15 & 0.79 & 1.69 \\
  $<$ 10th & Normal & 1,253 & 1.70 & 1.28 & 2.25 \\
  $<$ 10th & $<$ 10th & 1,531 & 3.83 & 3.21 & 4.58 \\
  \bottomrule
\end{tabular}
\end{table}

\begin{table}[h!]
\caption{The gestational-age-specific mortality risk and $95\%$ confidence interval for boys born preterm (22 to 29 weeks) with normal birthweight (BW) and head circumference (HC) are shown in bold font. The other rows show the mortality relative risk (as compared to infants with normal BW and HC) and $95\%$ confidence interval for infants with either one or both anthropometric measurements below the univariate 10th percentile. The sample size is denoted by $N$.\label{tab:B.2}}
\centering
\begin{tabular}{llrrrr}
  \toprule
\multicolumn{1}{l}{BW} & \multicolumn{1}{l}{HC} & \multicolumn{1}{c}{$N$} & \multicolumn{1}{c}{Risk} & \multicolumn{1}{c}{Lower} & \multicolumn{1}{c}{Upper} \\
  \midrule
  \multicolumn{6}{l}{$(21,23]$ weeks}\\
  \midrule
Normal & Normal & 5,161 & \textbf{0.59} & \textbf{0.58} & \textbf{0.60} \\
  Normal & $<$ 10th & 297 & 1.18 & 1.10 & 1.28 \\
  $<$ 10th & Normal & 351 & 1.29 & 1.21 & 1.37 \\
  $<$ 10th & $<$ 10th & 227 & 1.42 & 1.33 & 1.51 \\
  \midrule
  \multicolumn{6}{l}{$(23,25]$ weeks}\\
  \midrule
  Normal & Normal & 16,721 & \textbf{0.24} & \textbf{0.23} & \textbf{0.24} \\
  Normal & $<$ 10th & 601 & 1.44 & 1.29 & 1.62 \\
  $<$ 10th & Normal & 953 & 1.68 & 1.55 & 1.83 \\
  $<$ 10th & $<$ 10th & 1,030 & 2.36 & 2.22 & 2.50 \\
  \midrule
  \multicolumn{6}{l}{$(25,27]$ weeks}\\
  \midrule
  Normal & Normal & 21,623 & \textbf{0.08} & \textbf{0.08} & \textbf{0.09} \\
  Normal & $<$ 10th & 934 & 1.47 & 1.23 & 1.76 \\
  $<$ 10th & Normal & 954 & 2.08 & 1.79 & 2.40 \\
  $<$ 10th & $<$ 10th & 1,485 & 3.99 & 3.67 & 4.35 \\
  \midrule
  \multicolumn{6}{l}{$(27,29]$ weeks}\\
  \midrule
  Normal & Normal & 27,565 & \textbf{0.03} & \textbf{0.03} & \textbf{0.03} \\
  Normal & $<$ 10th & 1,247 & 1.63 & 1.26 & 2.11 \\
  $<$ 10th & Normal & 1,348 & 1.79 & 1.41 & 2.27 \\
  $<$ 10th & $<$ 10th & 1,810 & 4.52 & 3.94 & 5.18 \\
  \bottomrule
\end{tabular}
\end{table}

\begin{table}[ht]
\caption{The gestational-age-adjusted mortality risk and $95\%$ confidence interval for infants born preterm (22 to 29 weeks) with normal birthweight (BW) and head circumference (HC) are shown in bold font. The other rows show the mortality relative risk (as compared to infants with normal BW and HC) and $95\%$ confidence interval for infants with either one or both anthropometric measurements below the univariate 10th percentile. Estimates are given by sex, separately for infants whose HC-to-BW ratio in the isometric direction is below the 90th percentile (normal) or above it (supranormal). The sample size is denoted by $N$.\label{tab:B.3}}
\centering
\begin{tabular}{llrrrrrrrr}
  \toprule
  \multicolumn{1}{l}{BW} & \multicolumn{1}{l}{HC} & \multicolumn{1}{c}{$N$} & \multicolumn{1}{c}{Risk} & \multicolumn{1}{c}{Lower} & \multicolumn{1}{c}{Upper} & \multicolumn{1}{c}{$N$} & \multicolumn{1}{c}{Risk} & \multicolumn{1}{c}{Lower} & \multicolumn{1}{c}{Upper}\\
  \midrule
  \multicolumn{2}{l}{\textit{Females}} & \multicolumn{4}{c}{\textit{Normal HC-to-BW ratio}} & \multicolumn{4}{c}{\textit{Supranormal HC-to-BW ratio}}\\
  \midrule
Normal & Normal & 61,254 & \textbf{0.11} & \textbf{0.11} & \textbf{0.11} & 2,269 & \textbf{0.20} & \textbf{0.18} & \textbf{0.21} \\
  $<$ 10th & Normal & 405 & 0.93 & 0.70 & 1.25 & 2,750 & 0.91 & 0.61 & 1.37 \\
  $<$ 10th & $<$ 10th & 582 & 1.87 & 1.59 & 2.20 & 3,563 & 1.11 & 0.99 & 1.23 \\
  \midrule
  \multicolumn{2}{l}{\textit{Males}} & \multicolumn{4}{c}{\textit{Normal HC-to-BW ratio}} & \multicolumn{4}{c}{\textit{Supranormal HC-to-BW ratio}}\\
  \midrule
Normal & Normal & 70,223 & \textbf{0.13} & \textbf{0.13} & \textbf{0.14} & 847 & \textbf{0.26} & \textbf{0.23} & \textbf{0.29} \\
  $<$ 10th & Normal & 1,055 & 1.56 & 1.38 & 1.75 & 2,551 & 1.03 & 0.48 & 2.21 \\
  $<$ 10th & $<$ 10th & 1,098 & 1.86 & 1.68 & 2.07 & 3,454 & 1.02 & 0.89 & 1.16 \\
  \bottomrule
\end{tabular}
\end{table}


\begin{thebibliography}{51}

\bibitem[\protect\citeauthoryear{Azzalini and DallaValle}{1996}]{Azzalini1996}
\begin{barticle}[author]
\bauthor{\bsnm{Azzalini},~\bfnm{A}\binits{A.}} \AND
  \bauthor{\bsnm{DallaValle},~\bfnm{A}\binits{A.}}
(\byear{1996}).
\btitle{The multivariate skew-normal distribution}.
\bjournal{Biometrika}
\bvolume{83}
\bpages{715-726}.
\end{barticle}
\endbibitem

\bibitem[\protect\citeauthoryear{Bernstein et~al.}{2000}]{Bernstein2000}
\begin{barticle}[author]
\bauthor{\bsnm{Bernstein},~\bfnm{IM}\binits{I.}},
  \bauthor{\bsnm{Horbar},~\bfnm{JD}\binits{J.}},
  \bauthor{\bsnm{Badger},~\bfnm{GJ}\binits{G.}},
  \bauthor{\bsnm{Ohlsson},~\bfnm{A}\binits{A.}} \AND
  \bauthor{\bsnm{Golan},~\bfnm{A}\binits{A.}}
(\byear{2000}).
\btitle{Morbidity and mortality among very-low-birth-weight neonates with
  intrauterine growth restriction}.
\bjournal{American journal of Obstetrics and Gynecology}
\bvolume{182}
\bpages{198-206}.
\end{barticle}
\endbibitem

\bibitem[\protect\citeauthoryear{Bocca-Tjeertes et~al.}{2014}]{Bocca2014}
\begin{barticle}[author]
\bauthor{\bsnm{Bocca-Tjeertes},~\bfnm{I}\binits{I.}},
  \bauthor{\bsnm{Bos},~\bfnm{A}\binits{A.}},
  \bauthor{\bsnm{Kerstjens},~\bfnm{J}\binits{J.}}, \bauthor{\bsnm{{de
  Winter}},~\bfnm{A}\binits{A.}} \AND
  \bauthor{\bsnm{Reijneveld},~\bfnm{S}\binits{S.}}
(\byear{2014}).
\btitle{Symmetrical and asymmetrical growth restriction in preterm-born
  children}.
\bjournal{Pediatrics}
\bvolume{133}
\bpages{e650-e656}.
\end{barticle}
\endbibitem

\bibitem[\protect\citeauthoryear{Boghossian et~al.}{2016}]{Boghossian2016}
\begin{barticle}[author]
\bauthor{\bsnm{Boghossian},~\bfnm{NS}\binits{N.}},
  \bauthor{\bsnm{Geraci},~\bfnm{M}\binits{M.}},
  \bauthor{\bsnm{Edwards},~\bfnm{EM}\binits{E.}},
  \bauthor{\bsnm{Morrow},~\bfnm{KA}\binits{K.}} \AND
  \bauthor{\bsnm{Horbar},~\bfnm{JD}\binits{J.}}
(\byear{2016}).
\btitle{Anthropometric charts for infants born between 22 and 29
  weeks{\textquoteright} gestation}.
\bjournal{Pediatrics}.
\bnote{e20161641}.
\end{barticle}
\endbibitem

\bibitem[\protect\citeauthoryear{Boghossian et~al.}{2018a}]{Boghossian2018b}
\begin{barticle}[author]
\bauthor{\bsnm{Boghossian},~\bfnm{NS}\binits{N.}},
  \bauthor{\bsnm{Geraci},~\bfnm{M}\binits{M.}},
  \bauthor{\bsnm{Edwards},~\bfnm{EM}\binits{E.}} \AND
  \bauthor{\bsnm{Horbar},~\bfnm{JD}\binits{J.}}
(\byear{2018}a).
\btitle{Neonatal and fetal growth charts to identify preterm infants <30 weeks
  gestation at risk of adverse outcomes}.
\bjournal{American Journal of Obstetrics and Gynecology}
\bvolume{219}
\bpages{195.e1-195.e14}.
\end{barticle}
\endbibitem

\bibitem[\protect\citeauthoryear{Boghossian et~al.}{2018b}]{Boghossian2018}
\begin{barticle}[author]
\bauthor{\bsnm{Boghossian},~\bfnm{NS}\binits{N.}},
  \bauthor{\bsnm{Geraci},~\bfnm{M}\binits{M.}},
  \bauthor{\bsnm{Edwards},~\bfnm{EM}\binits{E.}} \AND
  \bauthor{\bsnm{Horbar},~\bfnm{JD}\binits{J.}}
(\byear{2018}b).
\btitle{Morbidity and mortality in small for gestational age infants at 22 to
  29 weeks' gestation}.
\bjournal{Pediatrics}.
\bnote{doi:10.1542/peds.2017-2533}.
\end{barticle}
\endbibitem

\bibitem[\protect\citeauthoryear{Brown et~al.}{1996}]{Brown1996}
\begin{barticle}[author]
\bauthor{\bsnm{Brown},~\bfnm{JE}\binits{J.}},
  \bauthor{\bsnm{Potter},~\bfnm{JD}\binits{J.}},
  \bauthor{\bsnm{Jacobs},~\bfnm{DR}\binits{D.}},
  \bauthor{\bsnm{Kopher},~\bfnm{RA}\binits{R.}},
  \bauthor{\bsnm{Rourke},~\bfnm{MJ}\binits{M.}},
  \bauthor{\bsnm{Barosso},~\bfnm{GM}\binits{G.}},
  \bauthor{\bsnm{Hannan},~\bfnm{PJ}\binits{P.}} \AND
  \bauthor{\bsnm{Schmid},~\bfnm{LA}\binits{L.}}
(\byear{1996}).
\btitle{Maternal waist-to-hip ratio as a predictor of newborn size: Results of
  the Diana project}.
\bjournal{Epidemiology}
\bvolume{7}
\bpages{62-66}.
\end{barticle}
\endbibitem

\bibitem[\protect\citeauthoryear{Chen et~al.}{2006}]{Chen2006}
\begin{barticle}[author]
\bauthor{\bsnm{Chen},~\bfnm{XK}\binits{X.}},
  \bauthor{\bsnm{Wen},~\bfnm{SW}\binits{S.}},
  \bauthor{\bsnm{Smith},~\bfnm{G}\binits{G.}},
  \bauthor{\bsnm{Yang},~\bfnm{Q}\binits{Q.}} \AND
  \bauthor{\bsnm{Walker},~\bfnm{M}\binits{M.}}
(\byear{2006}).
\btitle{General obstetrics: Pregnancy-induced hypertension is associated with
  lower infant mortality in preterm singletons}.
\bjournal{BJOG: An International Journal of Obstetrics \& Gynaecology}
\bvolume{113}
\bpages{544-551}.
\end{barticle}
\endbibitem

\bibitem[\protect\citeauthoryear{Corruccini}{1983}]{Corruccini1983}
\begin{barticle}[author]
\bauthor{\bsnm{Corruccini},~\bfnm{RS}\binits{R.}}
(\byear{1983}).
\btitle{Principal components for allometric analysis}.
\bjournal{American Journal of Physical Anthropology}
\bvolume{60}
\bpages{451-453}.
\end{barticle}
\endbibitem

\bibitem[\protect\citeauthoryear{Dashe et~al.}{2000}]{Dashe2000}
\begin{barticle}[author]
\bauthor{\bsnm{Dashe},~\bfnm{JS}\binits{J.}},
  \bauthor{\bsnm{McIntire},~\bfnm{DD}\binits{D.}},
  \bauthor{\bsnm{Lucas},~\bfnm{MJ}\binits{M.}} \AND
  \bauthor{\bsnm{Leveno},~\bfnm{KJ}\binits{K.}}
(\byear{2000}).
\btitle{Effects of symmetric and asymmetric fetal growth on pregnancy
  outcomes}.
\bjournal{Obstetrics \& Gynecology}
\bvolume{96}
\bpages{321-327}.
\end{barticle}
\endbibitem

\bibitem[\protect\citeauthoryear{Djeric et~al.}{2017}]{Djeric2017}
\begin{barticle}[author]
\bauthor{\bsnm{Djeric},~\bfnm{Mirjana}\binits{M.}},
  \bauthor{\bsnm{Ilincic},~\bfnm{Branislava}\binits{B.}},
  \bauthor{\bsnm{Cabarkapa},~\bfnm{Velibor}\binits{V.}},
  \bauthor{\bsnm{Radosavkic},~\bfnm{Isidora}\binits{I.}},
  \bauthor{\bsnm{Trifu},~\bfnm{Aleksandra}\binits{A.}} \AND
  \bauthor{\bsnm{Todorovic},~\bfnm{Masa}\binits{M.}}
(\byear{2017}).
\btitle{Prehypertension, waist-to-height ratio and markers of kidney function
  in apparently healthy men}.
\bjournal{Atherosclerosis}
\bvolume{263}
\bpages{e271}.
\end{barticle}
\endbibitem

\bibitem[\protect\citeauthoryear{Geraci}{2016}]{Geraci2016}
\begin{barticle}[author]
\bauthor{\bsnm{Geraci},~\bfnm{M}\binits{M.}}
(\byear{2016}).
\btitle{Estimation of regression quantiles in complex surveys with data missing
  at random: An application to birthweight determinants}.
\bjournal{Statistical Methods in Medical Research}
\bvolume{25}
\bpages{1393-1421}.
\end{barticle}
\endbibitem

\bibitem[\protect\citeauthoryear{Geraci, Alston and Birch}{2013}]{Geraci2013}
\begin{barticle}[author]
\bauthor{\bsnm{Geraci},~\bfnm{M}\binits{M.}},
  \bauthor{\bsnm{Alston},~\bfnm{RD}\binits{R.}} \AND
  \bauthor{\bsnm{Birch},~\bfnm{JM}\binits{J.}}
(\byear{2013}).
\btitle{Median percent change: A robust alternative for assessing temporal
  trends}.
\bjournal{Cancer Epidemiology}
\bvolume{37}
\bpages{843-849}.
\end{barticle}
\endbibitem

\bibitem[\protect\citeauthoryear{Geraci and Jones}{2015}]{Geraci2015}
\begin{barticle}[author]
\bauthor{\bsnm{Geraci},~\bfnm{M}\binits{M.}} \AND
  \bauthor{\bsnm{Jones},~\bfnm{MC}\binits{M.}}
(\byear{2015}).
\btitle{Improved transformation-based quantile regression}.
\bjournal{Canadian Journal of Statistics}
\bvolume{43}
\bpages{118-132}.
\end{barticle}
\endbibitem

\bibitem[\protect\citeauthoryear{Gilchrist}{2000}]{Gilchirst2000}
\begin{bbook}[author]
\bauthor{\bsnm{Gilchrist},~\bfnm{W}\binits{W.}}
(\byear{2000}).
\btitle{Statistical Modelling with Quantile Functions}.
\bpublisher{Chapman \& Hall/CRC}, \baddress{Boca Raton, FL}.
\end{bbook}
\endbibitem

\bibitem[\protect\citeauthoryear{Goodman and Kotz}{1973}]{Goodman1973}
\begin{barticle}[author]
\bauthor{\bsnm{Goodman},~\bfnm{IR}\binits{I.}} \AND
  \bauthor{\bsnm{Kotz},~\bfnm{S}\binits{S.}}
(\byear{1973}).
\btitle{Multivariate $\theta$-generalized normal distributions}.
\bjournal{Journal of Multivariate Analysis}
\bvolume{3}
\bpages{204-219}.
\end{barticle}
\endbibitem

\bibitem[\protect\citeauthoryear{Griffiths et~al.}{2007}]{Griffiths2007}
\begin{barticle}[author]
\bauthor{\bsnm{Griffiths},~\bfnm{L.~J.}\binits{L.~J.}},
  \bauthor{\bsnm{Dezateux},~\bfnm{C.}\binits{C.}},
  \bauthor{\bsnm{Cole},~\bfnm{T.~J.}\binits{T.~J.}} \AND
  \bauthor{\bparticle{the \uppercase{M}illennium \uppercase{C}ohort
  \uppercase{S}tudy \uppercase{C}hild~\uppercase{H}ealth}
  \bsnm{\uppercase{G}roup}}
(\byear{2007}).
\btitle{Differential parental weight and height contributions to offspring
  birthweight and weight gain in infancy}.
\bjournal{International Journal of Epidemiology}
\bvolume{36}
\bpages{104-107}.
\end{barticle}
\endbibitem

\bibitem[\protect\citeauthoryear{{Diverse Populations Collaborative
  Group}}{2005}]{DPCG2005}
\begin{barticle}[author]
\bauthor{\bsnm{{Diverse Populations Collaborative Group}}}
(\byear{2005}).
\btitle{Weight-height relationships and body mass index: Some observations from
  the diverse populations collaboration}.
\bjournal{American Journal of Physical Anthropology}
\bvolume{128}
\bpages{220-229}.
\end{barticle}
\endbibitem

\bibitem[\protect\citeauthoryear{Guellec et~al.}{2015}]{Guellec2015}
\begin{barticle}[author]
\bauthor{\bsnm{Guellec},~\bfnm{I}\binits{I.}},
  \bauthor{\bsnm{Marret},~\bfnm{S}\binits{S.}},
  \bauthor{\bsnm{Baud},~\bfnm{O}\binits{O.}},
  \bauthor{\bsnm{Cambonie},~\bfnm{G}\binits{G.}},
  \bauthor{\bsnm{Lapillonne},~\bfnm{A}\binits{A.}},
  \bauthor{\bsnm{Roze},~\bfnm{J-C}\binits{J.-C.}},
  \bauthor{\bsnm{Fresson},~\bfnm{J}\binits{J.}},
  \bauthor{\bsnm{Flamant},~\bfnm{C}\binits{C.}},
  \bauthor{\bsnm{Charkaluk},~\bfnm{M-L}\binits{M.-L.}},
  \bauthor{\bsnm{Arnaud},~\bfnm{C}\binits{C.}} \AND
  \bauthor{\bsnm{Ancel},~\bfnm{P-Y}\binits{P.-Y.}}
(\byear{2015}).
\btitle{Intrauterine growth restriction, head size at birth, and outcome in
  very preterm infants}.
\bjournal{The Journal of Pediatrics}
\bvolume{167}
\bpages{975-981.e2}.
\end{barticle}
\endbibitem

\bibitem[\protect\citeauthoryear{Hack et~al.}{1991}]{Hack1991}
\begin{barticle}[author]
\bauthor{\bsnm{Hack},~\bfnm{M}\binits{M.}},
  \bauthor{\bsnm{Breslau},~\bfnm{N}\binits{N.}},
  \bauthor{\bsnm{Weissman},~\bfnm{B}\binits{B.}},
  \bauthor{\bsnm{Aram},~\bfnm{Dorothy}\binits{D.}},
  \bauthor{\bsnm{Klein},~\bfnm{Nancy}\binits{N.}} \AND
  \bauthor{\bsnm{Borawski},~\bfnm{Elaine}\binits{E.}}
(\byear{1991}).
\btitle{Effect of very low birth weight and subnormal head size on cognitive
  abilities at school age}.
\bjournal{New England Journal of Medicine}
\bvolume{325}
\bpages{231-237}.
\end{barticle}
\endbibitem

\bibitem[\protect\citeauthoryear{Heymsfield et~al.}{2007}]{Heymsfield2007}
\begin{barticle}[author]
\bauthor{\bsnm{Heymsfield},~\bfnm{Steven~B.}\binits{S.~B.}},
  \bauthor{\bsnm{Gallagher},~\bfnm{Dympna}\binits{D.}},
  \bauthor{\bsnm{Mayer},~\bfnm{Laurel}\binits{L.}},
  \bauthor{\bsnm{Beetsch},~\bfnm{Joel}\binits{J.}} \AND
  \bauthor{\bsnm{Pietrobelli},~\bfnm{Angelo}\binits{A.}}
(\byear{2007}).
\btitle{Scaling of human body composition to stature: New insights into body
  mass index}.
\bjournal{The American Journal of Clinical Nutrition}
\bvolume{86}
\bpages{82-91}.
\end{barticle}
\endbibitem

\bibitem[\protect\citeauthoryear{Horbar et~al.}{2012}]{Horbar2012}
\begin{barticle}[author]
\bauthor{\bsnm{Horbar},~\bfnm{JD}\binits{J.}},
  \bauthor{\bsnm{Carpenter},~\bfnm{JH}\binits{J.}},
  \bauthor{\bsnm{Badger},~\bfnm{GJ}\binits{G.}},
  \bauthor{\bsnm{Kenny},~\bfnm{MJ}\binits{M.}},
  \bauthor{\bsnm{Soll},~\bfnm{RF}\binits{R.}},
  \bauthor{\bsnm{Morrow},~\bfnm{KA}\binits{K.}} \AND
  \bauthor{\bsnm{Buzas},~\bfnm{JS}\binits{J.}}
(\byear{2012}).
\btitle{Mortality and neonatal morbidity among infants 501 to 1500 grams from
  2000 to 2009}.
\bjournal{Pediatrics}
\bvolume{129}
\bpages{1019-1026}.
\end{barticle}
\endbibitem

\bibitem[\protect\citeauthoryear{Jain}{1997}]{Jain1997}
\begin{barticle}[author]
\bauthor{\bsnm{Jain},~\bfnm{L}\binits{L.}}
(\byear{1997}).
\btitle{Effect of pregnancy-induced and chronic hypertension on pregnancy
  outcome}.
\bjournal{Journal of Perinatology}
\bvolume{17}
\bpages{425—427}.
\end{barticle}
\endbibitem

\bibitem[\protect\citeauthoryear{Jolicoeur}{1963}]{Jolicoeur1963}
\begin{barticle}[author]
\bauthor{\bsnm{Jolicoeur},~\bfnm{P}\binits{P.}}
(\byear{1963}).
\btitle{193. Note: The multivariate generalization of the allometry equation}.
\bjournal{Biometrics}
\bvolume{19}
\bpages{497-499}.
\end{barticle}
\endbibitem

\bibitem[\protect\citeauthoryear{Judd, Ramirez-Zea and Stein}{2008}]{Judd2008}
\begin{barticle}[author]
\bauthor{\bsnm{Judd},~\bfnm{SE}\binits{S.}},
  \bauthor{\bsnm{Ramirez-Zea},~\bfnm{M}\binits{M.}} \AND
  \bauthor{\bsnm{Stein},~\bfnm{AD}\binits{A.}}
(\byear{2008}).
\btitle{Relation of ratio indices of anthropometric measures to obesity in a
  stunted population}.
\bjournal{American Journal of Human Biology}
\bvolume{20}
\bpages{446-450}.
\end{barticle}
\endbibitem

\bibitem[\protect\citeauthoryear{Koenker and Bassett}{1978}]{Koenker1978}
\begin{barticle}[author]
\bauthor{\bsnm{Koenker},~\bfnm{R}\binits{R.}} \AND
  \bauthor{\bsnm{Bassett},~\bfnm{G}\binits{G.}}
(\byear{1978}).
\btitle{Regression quantiles}.
\bjournal{Econometrica}
\bvolume{46}
\bpages{33-50}.
\end{barticle}
\endbibitem

\bibitem[\protect\citeauthoryear{Kong and Mizera}{2012}]{Kong2012}
\begin{barticle}[author]
\bauthor{\bsnm{Kong},~\bfnm{L}\binits{L.}} \AND
  \bauthor{\bsnm{Mizera},~\bfnm{I}\binits{I.}}
(\byear{2012}).
\btitle{Quantile tomography: Using quantiles with multivariate data}.
\bjournal{Statistica Sinica}
\bvolume{22}
\bpages{1589-1610}.
\end{barticle}
\endbibitem

\bibitem[\protect\citeauthoryear{Kozubowski and
  Podgorski}{2000}]{Kozubowski2000}
\begin{barticle}[author]
\bauthor{\bsnm{Kozubowski},~\bfnm{TJ}\binits{T.}} \AND
  \bauthor{\bsnm{Podgorski},~\bfnm{K}\binits{K.}}
(\byear{2000}).
\btitle{A multivariate and asymmetric generalization of Laplace distribution}.
\bjournal{Computational Statistics}
\bvolume{15}
\bpages{531-540}.
\end{barticle}
\endbibitem

\bibitem[\protect\citeauthoryear{Kozubowski, Podg\'{o}rski and
  Rychlik}{2013}]{Kozubowski2013}
\begin{barticle}[author]
\bauthor{\bsnm{Kozubowski},~\bfnm{TJ}\binits{T.}},
  \bauthor{\bsnm{Podg\'{o}rski},~\bfnm{K}\binits{K.}} \AND
  \bauthor{\bsnm{Rychlik},~\bfnm{I}\binits{I.}}
(\byear{2013}).
\btitle{Multivariate generalized Laplace distribution and related random
  fields}.
\bjournal{Journal of Multivariate Analysis}
\bvolume{113}
\bpages{59-72}.
\end{barticle}
\endbibitem

\bibitem[\protect\citeauthoryear{Kramer et~al.}{1990a}]{Kramer1990a}
\begin{barticle}[author]
\bauthor{\bsnm{Kramer},~\bfnm{MS}\binits{M.}},
  \bauthor{\bsnm{Olivier},~\bfnm{M}\binits{M.}},
  \bauthor{\bsnm{McLean},~\bfnm{FH}\binits{F.}},
  \bauthor{\bsnm{Dougherty},~\bfnm{GE}\binits{G.}},
  \bauthor{\bsnm{Willis},~\bfnm{DM}\binits{D.}} \AND
  \bauthor{\bsnm{Usher},~\bfnm{RH}\binits{R.}}
(\byear{1990}a).
\btitle{Determinants of fetal growth and body proportionality}.
\bjournal{Pediatrics}
\bvolume{86}
\bpages{18-26}.
\end{barticle}
\endbibitem

\bibitem[\protect\citeauthoryear{Kramer et~al.}{1990b}]{Kramer1990b}
\begin{barticle}[author]
\bauthor{\bsnm{Kramer},~\bfnm{MS}\binits{M.}},
  \bauthor{\bsnm{Olivier},~\bfnm{M}\binits{M.}},
  \bauthor{\bsnm{McLean},~\bfnm{FH}\binits{F.}},
  \bauthor{\bsnm{Willis},~\bfnm{DM}\binits{D.}} \AND
  \bauthor{\bsnm{Usher},~\bfnm{RH}\binits{R.}}
(\byear{1990}b).
\btitle{Impact of intrauterine growth retardation and body proportionality on
  fetal and neonatal outcome}.
\bjournal{Pediatrics}
\bvolume{86}
\bpages{707-713}.
\end{barticle}
\endbibitem

\bibitem[\protect\citeauthoryear{Kuban et~al.}{2009}]{Kuban2009}
\begin{barticle}[author]
\bauthor{\bsnm{Kuban},~\bfnm{KCK}\binits{K.}},
  \bauthor{\bsnm{Allred},~\bfnm{EN}\binits{E.}},
  \bauthor{\bsnm{O'Shea},~\bfnm{TM}\binits{T.}},
  \bauthor{\bsnm{Paneth},~\bfnm{N}\binits{N.}},
  \bauthor{\bsnm{Westra},~\bfnm{S}\binits{S.}},
  \bauthor{\bsnm{Miller},~\bfnm{C}\binits{C.}},
  \bauthor{\bsnm{Rosman},~\bfnm{NP}\binits{N.}} \AND
  \bauthor{\bsnm{Leviton},~\bfnm{A}\binits{A.}}
(\byear{2009}).
\btitle{Developmental correlates of head circumference at birth and two years
  in a cohort of extremely low gestational age newborns}.
\bjournal{The Journal of Pediatrics}
\bvolume{155}
\bpages{344-349.e3}.
\end{barticle}
\endbibitem

\bibitem[\protect\citeauthoryear{Lebowitz and Blumenthal}{1993}]{Lebowitz1993}
\begin{barticle}[author]
\bauthor{\bsnm{Lebowitz},~\bfnm{MR}\binits{M.}} \AND
  \bauthor{\bsnm{Blumenthal},~\bfnm{SA}\binits{S.}}
(\byear{1993}).
\btitle{The molar ratio of insulin to {C}-peptide: An aid to the diagnosis of
  hypoglycemia due to surreptitious (or inadvertent) insulin administration}.
\bjournal{Archives of Internal Medicine}
\bvolume{153}
\bpages{650-655}.
\end{barticle}
\endbibitem

\bibitem[\protect\citeauthoryear{Lin, Su and River}{1991}]{Lin1991}
\begin{barticle}[author]
\bauthor{\bsnm{Lin},~\bfnm{C-C}\binits{C.-C.}},
  \bauthor{\bsnm{Su},~\bfnm{S-J}\binits{S.-J.}} \AND
  \bauthor{\bsnm{River},~\bfnm{LP}\binits{L.}}
(\byear{1991}).
\btitle{Comparison of associated high-risk factors and perinatal outcome
  between symmetric and asymmetric fetal intrauterine growth retardation}.
\bjournal{American Journal of Obstetrics \& Gynecology}
\bvolume{164}
\bpages{1535-1542}.
\end{barticle}
\endbibitem

\bibitem[\protect\citeauthoryear{Mao, Wei and Liu}{2017}]{Mao2017}
\begin{barticle}[author]
\bauthor{\bsnm{Mao},~\bfnm{G}\binits{G.}},
  \bauthor{\bsnm{Wei},~\bfnm{Y}\binits{Y.}} \AND
  \bauthor{\bsnm{Liu},~\bfnm{Y}\binits{Y.}}
(\byear{2017}).
\btitle{{SIMEX} method for censored quantile regression with measurement
  error}.
\bjournal{Communications in Statistics - Simulation and Computation}
\bvolume{46}
\bpages{7552-7560}.
\end{barticle}
\endbibitem

\bibitem[\protect\citeauthoryear{McBride et~al.}{2017}]{McBride2017}
\begin{barticle}[author]
\bauthor{\bsnm{McBride},~\bfnm{CA}\binits{C.}},
  \bauthor{\bsnm{Bernstein},~\bfnm{IM}\binits{I.}},
  \bauthor{\bsnm{Badger},~\bfnm{GJ}\binits{G.}} \AND
  \bauthor{\bsnm{Soll},~\bfnm{RF}\binits{R.}}
(\byear{2017}).
\btitle{Maternal hypertension and mortality in small for gestational age 22- to
  29-week infants}.
\bjournal{Reproductive Sciences}
\bvolume{25}
\bpages{276-280}.
\end{barticle}
\endbibitem

\bibitem[\protect\citeauthoryear{McKeague et~al.}{2011}]{McKeague2011}
\begin{barticle}[author]
\bauthor{\bsnm{McKeague},~\bfnm{IW}\binits{I.}},
  \bauthor{\bsnm{L\'{o}pez-Pintado},~\bfnm{S}\binits{S.}},
  \bauthor{\bsnm{Hallin},~\bfnm{M}\binits{M.}} \AND
  \bauthor{\bsnm{\v{S}iman},~\bfnm{M}\binits{M.}}
(\byear{2011}).
\btitle{Analyzing growth trajectories}.
\bjournal{Journal of Developmental Origins of Health and Disease}
\bvolume{2}
\bpages{322-329}.
\end{barticle}
\endbibitem

\bibitem[\protect\citeauthoryear{Molnos et~al.}{2018}]{Molnos2018}
\begin{barticle}[author]
\bauthor{\bsnm{Molnos},~\bfnm{S}\binits{S.}},
  \bauthor{\bsnm{Wahl},~\bfnm{S}\binits{S.}},
  \bauthor{\bsnm{Haid},~\bfnm{M}\binits{M.}},
  \bauthor{\bsnm{Eekhoff},~\bfnm{EMW}\binits{E.}},
  \bauthor{\bsnm{Pool},~\bfnm{R}\binits{R.}},
  \bauthor{\bsnm{Floegel},~\bfnm{Anna}\binits{A.}},
  \bauthor{\bsnm{Deelen},~\bfnm{Joris}\binits{J.}},
  \bauthor{\bsnm{Much},~\bfnm{Daniela}\binits{D.}},
  \bauthor{\bsnm{Prehn},~\bfnm{Cornelia}\binits{C.}},
  \bauthor{\bsnm{Breier},~\bfnm{Michaela}\binits{M.}},
  \bauthor{\bsnm{Draisma},~\bfnm{Harmen~H.}\binits{H.~H.}},
  \bauthor{\bparticle{van} \bsnm{Leeuwen},~\bfnm{Nienke}\binits{N.}},
  \bauthor{\bsnm{Simonis-Bik},~\bfnm{Annemarie M.~C.}\binits{A.~M.~C.}},
  \bauthor{\bsnm{Jonsson},~\bfnm{Anna}\binits{A.}},
  \bauthor{\bsnm{Willemsen},~\bfnm{Gonneke}\binits{G.}},
  \bauthor{\bsnm{Bernigau},~\bfnm{Wolfgang}\binits{W.}},
  \bauthor{\bsnm{Wang-Sattler},~\bfnm{Rui}\binits{R.}},
  \bauthor{\bsnm{Suhre},~\bfnm{Karsten}\binits{K.}},
  \bauthor{\bsnm{Peters},~\bfnm{Annette}\binits{A.}},
  \bauthor{\bsnm{Thorand},~\bfnm{Barbara}\binits{B.}},
  \bauthor{\bsnm{Herder},~\bfnm{Christian}\binits{C.}},
  \bauthor{\bsnm{Rathmann},~\bfnm{Wolfgang}\binits{W.}},
  \bauthor{\bsnm{Roden},~\bfnm{Michael}\binits{M.}},
  \bauthor{\bsnm{Gieger},~\bfnm{Christian}\binits{C.}},
  \bauthor{\bsnm{Kramer},~\bfnm{Mark H.~H.}\binits{M.~H.~H.}},
  \bauthor{\bparticle{van} \bsnm{Heemst},~\bfnm{Diana}\binits{D.}},
  \bauthor{\bsnm{Pedersen},~\bfnm{Helle~K.}\binits{H.~K.}},
  \bauthor{\bsnm{Gudmundsdottir},~\bfnm{Valborg}\binits{V.}},
  \bauthor{\bsnm{Schulze},~\bfnm{Matthias~B.}\binits{M.~B.}},
  \bauthor{\bsnm{Pischon},~\bfnm{Tobias}\binits{T.}}, \bauthor{\bparticle{de}
  \bsnm{Geus},~\bfnm{Eco J.~C.}\binits{E.~J.~C.}},
  \bauthor{\bsnm{Boeing},~\bfnm{Heiner}\binits{H.}},
  \bauthor{\bsnm{Boomsma},~\bfnm{Dorret~I.}\binits{D.~I.}},
  \bauthor{\bsnm{Ziegler},~\bfnm{Anette~G.}\binits{A.~G.}},
  \bauthor{\bsnm{Slagboom},~\bfnm{P.~Eline}\binits{P.~E.}},
  \bauthor{\bsnm{Hummel},~\bfnm{Sandra}\binits{S.}},
  \bauthor{\bsnm{Beekman},~\bfnm{Marian}\binits{M.}},
  \bauthor{\bsnm{Grallert},~\bfnm{Harald}\binits{H.}},
  \bauthor{\bsnm{Brunak},~\bfnm{Søren}\binits{S.}},
  \bauthor{\bsnm{McCarthy},~\bfnm{Mark~I.}\binits{M.~I.}},
  \bauthor{\bsnm{Gupta},~\bfnm{Ramneek}\binits{R.}},
  \bauthor{\bsnm{Pearson},~\bfnm{Ewan~R.}\binits{E.~R.}},
  \bauthor{\bsnm{Adamski},~\bfnm{Jerzy}\binits{J.}} \AND
  \bauthor{\bparticle{’t} \bsnm{Hart},~\bfnm{Leen~M.}\binits{L.~M.}}
(\byear{2018}).
\btitle{Metabolite ratios as potential biomarkers for type 2 diabetes: a DIRECT
  study}.
\bjournal{Diabetologia}
\bvolume{61}
\bpages{117-129}.
\end{barticle}
\endbibitem

\bibitem[\protect\citeauthoryear{O'Brien}{2015}]{OBrien2015}
\begin{barticle}[author]
\bauthor{\bsnm{O'Brien},~\bfnm{DM}\binits{D.}}
(\byear{2015}).
\btitle{Stable isotope ratios as biomarkers of diet for health research}.
\bjournal{Annual Review of Nutrition}
\bvolume{35}
\bpages{565-594}.
\end{barticle}
\endbibitem

\bibitem[\protect\citeauthoryear{Olsen et~al.}{2009}]{Olsen2009}
\begin{barticle}[author]
\bauthor{\bsnm{Olsen},~\bfnm{IE}\binits{I.}},
  \bauthor{\bsnm{Lawson},~\bfnm{ML}\binits{M.}},
  \bauthor{\bsnm{Meinzen-Derr},~\bfnm{J}\binits{J.}},
  \bauthor{\bsnm{Sapsford},~\bfnm{AL}\binits{A.}},
  \bauthor{\bsnm{Schibler},~\bfnm{KR}\binits{K.}},
  \bauthor{\bsnm{Donovan},~\bfnm{EF}\binits{E.}} \AND
  \bauthor{\bsnm{Morrow},~\bfnm{AL}\binits{A.}}
(\byear{2009}).
\btitle{Use of a body proportionality index for growth assessment of preterm
  infants}.
\bjournal{The Journal of Pediatrics}
\bvolume{154}
\bpages{486-491}.
\end{barticle}
\endbibitem

\bibitem[\protect\citeauthoryear{Quetelet}{1842}]{Quetelet1842}
\begin{bbook}[author]
\bauthor{\bsnm{Quetelet},~\bfnm{MA}\binits{M.}}
(\byear{1842}).
\btitle{A treatise on man and the development of his faculties}.
\bpublisher{W. and R. Chambers}, \baddress{Edinburgh}.
\end{bbook}
\endbibitem

\bibitem[\protect\citeauthoryear{Rotondi et~al.}{2018}]{Rotondi2018}
\begin{barticle}[author]
\bauthor{\bsnm{Rotondi},~\bfnm{S}\binits{S.}},
  \bauthor{\bsnm{Tartaglione},~\bfnm{L}\binits{L.}},
  \bauthor{\bsnm{Muci},~\bfnm{ML}\binits{M.}},
  \bauthor{\bsnm{Farcomeni},~\bfnm{A}\binits{A.}},
  \bauthor{\bsnm{Pasquali},~\bfnm{M}\binits{M.}} \AND
  \bauthor{\bsnm{Mazzaferro},~\bfnm{S}\binits{S.}}
(\byear{2018}).
\btitle{Oxygen extraction ratio (OER) as a measurement of hemodialysis (HD)
  induced tissue hypoxia: A pilot study}.
\bjournal{Scientific Reports}
\bvolume{8}
\bpages{5655}.
\end{barticle}
\endbibitem

\bibitem[\protect\citeauthoryear{Saleem et~al.}{2011}]{Saleem2011}
\begin{barticle}[author]
\bauthor{\bsnm{Saleem},~\bfnm{T}\binits{T.}},
  \bauthor{\bsnm{Sajjad},~\bfnm{N}\binits{N.}},
  \bauthor{\bsnm{Fatima},~\bfnm{S}\binits{S.}},
  \bauthor{\bsnm{Habib},~\bfnm{N}\binits{N.}},
  \bauthor{\bsnm{Ali},~\bfnm{SR}\binits{S.}} \AND
  \bauthor{\bsnm{Qadir},~\bfnm{M}\binits{M.}}
(\byear{2011}).
\btitle{Intrauterine growth retardation--small events, big consequences}.
\bjournal{Italian Journal of Pediatrics}
\bvolume{37}
\bpages{41}.
\end{barticle}
\endbibitem

\bibitem[\protect\citeauthoryear{Somers}{1986}]{Somers1986}
\begin{barticle}[author]
\bauthor{\bsnm{Somers},~\bfnm{KM}\binits{K.}}
(\byear{1986}).
\btitle{Multivariate allometry and removal of size with principal components
  analysis}.
\bjournal{Systematic Zoology}
\bvolume{35}
\bpages{359-368}.
\end{barticle}
\endbibitem

\bibitem[\protect\citeauthoryear{Stoll et~al.}{2010}]{Stoll2010}
\begin{barticle}[author]
\bauthor{\bsnm{Stoll},~\bfnm{BJ}\binits{B.}},
  \bauthor{\bsnm{Hansen},~\bfnm{NI}\binits{N.}},
  \bauthor{\bsnm{Bell},~\bfnm{EF}\binits{E.}},
  \bauthor{\bsnm{Shankaran},~\bfnm{S}\binits{S.}},
  \bauthor{\bsnm{Laptook},~\bfnm{AR}\binits{A.}},
  \bauthor{\bsnm{Walsh},~\bfnm{MC}\binits{M.}},
  \bauthor{\bsnm{Hale},~\bfnm{EC}\binits{E.}},
  \bauthor{\bsnm{Newman},~\bfnm{NS}\binits{N.}},
  \bauthor{\bsnm{Schibler},~\bfnm{K}\binits{K.}} \AND
  \bauthor{\bsnm{Carlo},~\bfnm{WA}\binits{W.}}
(\byear{2010}).
\btitle{Neonatal outcomes of extremely preterm infants from the {NICHD Neonatal
  Research Network}}.
\bjournal{Pediatrics}
\bvolume{126}
\bpages{443-456}.
\end{barticle}
\endbibitem

\bibitem[\protect\citeauthoryear{Vandenbosche and
  Kirchner}{1998}]{Vandenbosche1998}
\begin{barticle}[author]
\bauthor{\bsnm{Vandenbosche},~\bfnm{RC}\binits{R.}} \AND
  \bauthor{\bsnm{Kirchner},~\bfnm{JT}\binits{J.}}
(\byear{1998}).
\btitle{Intrauterine growth retardation}.
\bjournal{American Family Physician}
\bvolume{58}
\bpages{1384-90, 1393-4}.
\end{barticle}
\endbibitem

\bibitem[\protect\citeauthoryear{Warton et~al.}{2006}]{Warton2006}
\begin{barticle}[author]
\bauthor{\bsnm{Warton},~\bfnm{DI}\binits{D.}},
  \bauthor{\bsnm{Wright},~\bfnm{IJ}\binits{I.}},
  \bauthor{\bsnm{Falster},~\bfnm{DS}\binits{D.}} \AND
  \bauthor{\bsnm{Westoby},~\bfnm{M}\binits{M.}}
(\byear{2006}).
\btitle{Bivariate line-fitting methods for allometry}.
\bjournal{Biological Reviews}
\bvolume{81}
\bpages{259-291}.
\end{barticle}
\endbibitem

\bibitem[\protect\citeauthoryear{Warton et~al.}{2012}]{Warton2012}
\begin{barticle}[author]
\bauthor{\bsnm{Warton},~\bfnm{DI}\binits{D.}},
  \bauthor{\bsnm{Duursma},~\bfnm{RA}\binits{R.}},
  \bauthor{\bsnm{Falster},~\bfnm{DS}\binits{D.}} \AND
  \bauthor{\bsnm{Taskinen},~\bfnm{S}\binits{S.}}
(\byear{2012}).
\btitle{{smatr 3} -- an {R} package for estimation and inference about
  allometric lines}.
\bjournal{Methods in Ecology and Evolution}
\bvolume{3}
\bpages{257-259}.
\end{barticle}
\endbibitem

\bibitem[\protect\citeauthoryear{Wei}{2008}]{Wei2008}
\begin{barticle}[author]
\bauthor{\bsnm{Wei},~\bfnm{Y}\binits{Y.}}
(\byear{2008}).
\btitle{An approach to multivariate covariate-dependent quantile contours with
  application to bivariate conditional growth charts}.
\bjournal{Journal of the American Statistical Association}
\bvolume{103}
\bpages{397-409}.
\end{barticle}
\endbibitem

\bibitem[\protect\citeauthoryear{Wei and Carroll}{2009}]{Wei2012}
\begin{barticle}[author]
\bauthor{\bsnm{Wei},~\bfnm{Y}\binits{Y.}} \AND
  \bauthor{\bsnm{Carroll},~\bfnm{RJ}\binits{R.}}
(\byear{2009}).
\btitle{Quantile regression with measurement error}.
\bjournal{Journal of the American Statistical Association}
\bvolume{104}
\bpages{1129-1143}.
\end{barticle}
\endbibitem

\bibitem[\protect\citeauthoryear{Williams and O'Brien}{1998}]{Williams1998}
\begin{barticle}[author]
\bauthor{\bsnm{Williams},~\bfnm{MC}\binits{M.}} \AND
  \bauthor{\bsnm{O'Brien},~\bfnm{WF}\binits{W.}}
(\byear{1998}).
\btitle{Low weight/length ratio to assess risk of cerebral palsy and perinatal
  mortality in twins}.
\bjournal{American Journal of Perinatology}
\bvolume{11}
\bpages{225-228}.
\end{barticle}
\endbibitem

\end{thebibliography}
\end{document}